\documentclass[]{sn-jnl} 

\usepackage{graphicx}%
\usepackage{multirow}%
\usepackage{amsmath,amssymb,amsfonts}%
\usepackage{amsthm}%
\usepackage{mathrsfs}%
\usepackage[title]{appendix}%
\usepackage{xcolor}%
\usepackage{textcomp}%
\usepackage{manyfoot}%
\usepackage{booktabs}%
\usepackage{algorithm}%
\usepackage{algorithmicx}%
\usepackage{algpseudocode}%
\usepackage{listings}%
\usepackage[numbers,sort&compress]{natbib}
\usepackage{tikz}
\usetikzlibrary{arrows.meta, positioning}
\usepackage{lmodern}
\usepackage{anyfontsize}
\usepackage{euscript}
\usepackage{subfigure, bm}
\usepackage{bbm}
\usepackage[english]{babel}
\usepackage{enumitem} 
\usepackage{soul}
\usepackage{mathtools}
\mathtoolsset{showonlyrefs}
\usepackage{calligra}
\DeclareMathAlphabet{\mathcalligra}{T1}{calligra}{m}{n}
\usepackage{color}

\newcommand{\E}{\mathbb{E}}

\newcommand{\triplenorm}[1]{{\vert\kern-0.25ex\vert\kern-0.25ex\vert #1 
    \vert\kern-0.25ex\vert\kern-0.25ex\vert}}

\newtheorem{thm}{Theorem}[section]
\newtheorem{prop}{Proposition}[section]
\newtheorem{lemma}{Lemma}[section]

\theoremstyle{definition}
\newtheorem{definition}{Definition}[section]
\newtheorem{asmp}{Assumption}
\newtheorem{example}{Example}[section]
\newtheorem{remark}{Remark}

\newcommand{\xdownarrow}[1]{%
  {\left\downarrow\vbox to #1{}\right.\kern-\nulldelimiterspace}
}

\begin{document}

\title{Relative Arbitrage Opportunities in an Extended Mean Field System}

\author*[1]{\fnm{Nicole Tianjiao} \sur{Yang}}\email{\href{mailto:nicoley031@gmail.com}{nicole.yang@utk.edu}}

\author[2]{\fnm{Tomoyuki} \sur{Ichiba}}\email{\href{mailto:ichiba@pstat.ucsb.edu}{ichiba@pstat.ucsb.edu}}

\affil*[1]{\orgdiv{Department of Mathematics}, \orgname{University of Tennessee, Knoxville}, \orgaddress{\postcode{37996}, \state{TN}, \country{USA}}}

\affil[2]{\orgdiv{Department of Statistics and Applied Probability}, \orgname{University of California, Santa Barbara}, \orgaddress{\street{South Hall}, \city{Santa Barbara}, \postcode{93106}, \state{CA}, \country{USA}}}

\abstract{
This paper studies relative arbitrage opportunities in a market with competitive investors through stochastic differential games in the limit as the number of players tends to infinity. {With common noises introduced by the stock capitalization processes, we establish a conditional McKean-Vlasov system to study the market dynamics coupled to the expected trading volume of investors. We show that optimal arbitrage can be characterized as a solution of a Cauchy PDE constructed by the volatility terms in the market model. The structure of the market dynamics can be relaxed, and we provide a theoretical framework to study a general mean-field system, where the interaction is characterized by a joint distribution of wealth and strategies. In this setting, the optimal relative arbitrage constitutes the strong equilibrium of an extended mean-field game. We provide conditions for the existence and uniqueness of the mean-field equilibrium. We further prove the propagation of chaos result for the finite-player game counterpart, and demonstrate that the Nash equilibrium converges to the mean field equilibrium when the population grows to infinity.}}




\maketitle

\section{Introduction}

We consider a large number of non-cooperative active money managers with interactions, each of whom aims to outperform a benchmark portfolio over a given time horizon. This outperformance, termed relative arbitrage, is a significant topic in Stochastic Portfolio Theory (SPT)  \cite{fernholz2002stochastic}. 
The optimal arbitrage and hedging opportunities in the complete market are studied in \cite{fernholz2010optimal, fernholz2011optimal, ruf2011optimal}. Functionally generated portfolios introduced in SPT are used to construct portfolios with favored return characteristics. 
It guarantees the relative return of portfolios generated by appropriate deterministic functions of the market portfolio. 
The optimization problem from the functional generated portfolio point of view is handled in \cite{wong2015optimization}. The paper \cite{wong} connects functional portfolio generation with information theory and relative arbitrage.


{
In this work, we investigate the relative arbitrage problem in a mean-field game (\cite{lasry2007mean, huang2006large}) regime and study the limiting behavior of the market and investors when the number of investors tends to infinity.  Each investor has his or her own performance criterion, by which their growth rate should exceed the benchmark with probability one that the investor considers the investment successful. The previous work \cite{ichiba2020relative} investigates optimal arbitrage opportunities when there are interactions among a large, finite number of investors. Despite a well-posed interacting dynamical system and the Nash equilibrium established, it has been demonstrated in the work that, due to the complicated dynamical system, achieving the uniqueness of Nash equilibrium requires a sufficiently large number of players in such stochastic games, or a short time horizon. Meanwhile, it is unlikely to achieve a tractable equilibrium, especially when the number of players $N$ is large.}


{
Thus, we discuss the optimization of the relative arbitrage quantity and the corresponding strategies considering the formulation of mean-field games with common noise in the market, and justify the mean-field approximation to a finite-player relative arbitrage problem (\cite{ichiba2020relative}). In particular, the arbitrage we consider is relative to a benchmark that includes the market capitalization and wealth of a continuum of investors. Hence, the game formulation is of the extended mean field game type \cite[Volume I, Section 4.6]{carmona2018probabilistic}, as the interaction among investors is through a joint distribution of their investment strategies and wealth. We show that the optimal arbitrage function, which is the minimal nonnegative solution of an associated Cauchy PDE, acts as the value function of the mean-field game we consider. We then derive the propagation of chaos results so that a player in a large game limit feels the presence of other players through the statistical distribution of states and actions. The equilibria of the $N$-player games can be shown to converge to the MFG limit. More importantly, this can be used as an approximation of the realistic finite player game, and shed light on the information of a complex market system.}

There are several reasons for this work that differentiate it from a traditional mean-field game setup. Firstly, instead of solving the value function at equilibrium through a (stochastic) HJB equation, or construct a system of FBSDEs, we characterize the optimal value function as the solution of the minimal nonnegative solution of a Cauchy PDE. The coefficient in the Cauchy PDE is still coupled with the wealth and strategies of investors. Secondly, in the spirit of functional generated portfolio, we derive that the uniqueness of mean field equilibrium depends on the expectation of arbitrage preference and the weight of market capitalization in the relative benchmark. Third, instead of perfectly symmetric stochastic games, we differentiate the players by their preference to outperform the relative benchmark. 

Finally, we elucidate the approximation between the finite and the limiting systems in the following diagram.
\[
\text{Market dynamics} \hspace{1.5cm} \text{Relative arbitrage of $N$ investors} 
\]
\[
\bigg\downarrow \hspace{5cm} \quad \bigg\downarrow 
\]
\[
N-\text{particle dynamics} \overset{\cite{ichiba2020relative}}{\xrightarrow{\hspace{1.7cm}}} N \text{-player Nash Equilibrium} 
\]
\[
\overset{Proposition~\ref{xlimit}-\ref{prop: tightv}}{\bigg\downarrow} \hspace{3.2cm} \quad  \overset{Theorem~\ref{tight}-\ref{netomfe}}{\bigg\uparrow} 
\quad \quad
\]
\[
\text{$\infty$-particle dynamics} \overset{Theorem~\ref{mfeuniq} }{\xrightarrow{\hspace{1.5cm}}} \text{Mean Field Equilibrium} 
\]
\bigskip

\subsection{Related work}

\noindent \textbf{Relative arbitrage and functionally generated portfolio} The optimal arbitrage and hedging opportunities of a single investor in a complete market are studied in \cite{fernholz2010optimal, fernholz2011optimal, ruf2011optimal}. The previous work \cite{ichiba2020relative} follows this framework and extends the optimization problem to a market with a given finite number of investors interacted through the equity market.
From the perspective of functionally generated portfolios, \cite{wong2015optimization, campbell2022functional} and the references therein consider relative arbitrage opportunities by optimizing the functionally generated portfolios. In particular, \cite{wong2015universal, cuchiero2019cover} extends Cover's universal portfolio (\cite{cover1991universal}) approach to families of functionally generated portfolios. All of these works benefit from their constructions of the optimization problem so that no drift estimation is required. This is opposed to classical portfolio optimization methods, such as expected utility maximization. We demonstrate the connection of the above two directions to find optimal relative arbitrage in this work.

\noindent \textbf{Mean field games and controls} Mean field games are grounded in the premise of the convergence of Nash equilibria of large stochastic differential games with a mean-field interaction. See \cite{carmona2018probabilistic} and the references therein for the theory and applications of mean field games. A player in a large game limit should feel the presence of other players through the statistical distribution of states and actions. In the mean field game formulation, we pursue the equilibrium for a so-called representative player, which is a generic player in the infinite-population. For this reason, the MFG framework is expected to be more tractable than $N$-player games. This is one of our motivations to study relative arbitrage opportunities in the infinite-investor regime. Among the fruitful developments of mean field theory and its applications, the mean field game with common noise, or the extended mean field games, remains a challenging topic. The input of an extended mean field game is parameterized by an external random environment. Recent developments such as \cite{carmona2015probabilistic,carmona2016mean, djete2022extended, djete2023mean, 10.1214/22-AAP1876} focus on mean field games and controls with common noise and construct a weak notion of mean field equilibrium. The mean field games with common noise but without idiosyncratic noise are studied by \cite{cardaliaguet2022first} from the PDE perspective. A recent work \cite{mou2023minimal} proposes a partial order for the set of measure flows when the monotonicity condition is not satisfied for the extended mean field game. In this setup, the minimal and maximal mean-field equilibria are constructed. The regime of $N$-player and common noise mean field games has been used in \cite{lacker2019mean, hu2022n} to study portfolio optimization, with players subject to the CARA or CARR utilities. 


\subsection{Organization}
The organization of the paper is as follows. Section~\ref{fds} introduces the market setup, dynamics of the investors as well as the mean field interactions. Section~\ref{sec: ra} discusses the optimization of relative arbitrage opportunities under the Markovian market model. {Section~\ref{sec: mfg} establishes the optimization of relative arbitrage as searching for Nash equilibrium in extended mean-field games. Section~\ref{lims} constructs a conditional McKean-Vlasov SDE of the form that the coefficients of diffusion depend on the joint distribution of the state processes and the control, and shows the propagation of chaos holds to provide proofs of the representative agent model used in previous sections. We justify that the mean-field formulation is an appropriate generalization of the $N$-player relative arbitrage problem. In Section~\ref{sec: discussion}, we discuss potential topics that can be extended from this paper. Some proofs are placed in Appendix~\ref{sec:appendix}. }

\subsection{Notations}
{
For easy reference, we summarize the following important quantities in this paper.
\begin{table}[h!]
\centering
\begin{tabular}{lll}
\toprule
\textbf{Notation} & \textbf{Description} & \textbf{Defined in}\\
\midrule
$n$ & Total number of stocks in the equity market & Section~\ref{sec: cap}\\
$V^\ell$ & Wealth process of investor $\ell$ & Section~\ref{sec:investors}\\
$\pi^{\ell}$ & Strategy profile of investor $\ell$& Section~\ref{sec:investors}\\
$(\mathcal{X},\mathcal{Z}), \ (\mathbf{x}, \mathbf{z})$ & Interacting market system and its initial condition & Section~\ref{sec: marketz}\\
$\theta, \ L$ & Market price of risk, Stochastic discount factor & Section~\ref{sec: marketz}\\
$\mathcal V$ & Benchmark that investor compare their wealth to & Section~\ref{sec: benchmark}\\
$u^{\ell}$ & Optimal arbitrage quantity of investor $\ell$ & Section~\ref{sec: u}\\
$J^{\ell}$, $J$ & Cost functional of investor $\ell$ and of the representative player & Section~\ref{exdmfg}\\
$\mu, \ \nu$ & Path measures that represent mean field interactions & Section~\ref{exdmfg} \& \ref{lims} \\
$\Phi$ & The best response map of strategy trajectories& Section~\ref{sec: solveNE}\\
Superscript $N$ & Denotes the counterparts in the $N$-player game & Section~\ref{limitsection}\\
\bottomrule
\end{tabular}
\caption{Glossary of main notations, brief description of their meanings, and where each term is first defined.}
\end{table}
}

\section{The Market Model}
\label{fds}
\numberwithin{equation}{section}
We consider an equity market and focus on the market dynamics and behaviors of a group of investors in this market. The number of investors that we include is large enough to affect the market as a whole.
\subsection{Capitalizations}
\label{sec: cap}
For a given finite time horizon $[0,T]$, an admissible market model $\mathcal{M}$ we use in this paper consists of a given $n$ dimensional standard Brownian motion $B(\cdot) := (B_1(\cdot), \ldots, B_n(\cdot))^{\prime}$ on a canonical filtered probability space $(\Omega, \mathcal{F}, \mathbb{F}, \mathbb{P})$. 
Filtration $\mathbb{F}$ represents the ``flow of information'' in the market, where
{$\mathbb{F} = \{\mathcal{F}^B(t)\}_{0 \leq t < \infty}$} 
and $\mathcal{F}^B(t) := \{\sigma(B(s)) ; 0 < s < t\}_{0 \leq t < \infty} $ with $\mathcal{F}^B(0) := \{\emptyset, \Omega\}$, mod $\mathbb{P}$. All local martingales and supermartingales are with respect to the filtration $\mathbb{F}$ if not written specifically.
Thus, there are $n$ risky assets (stocks) with prices per share $\mathcal{X}(\cdot) = (X_1(\cdot), \ldots, X_n(\cdot))^{\prime}$ driven by $n$ independent Brownian motions in the following system of stochastic differential equations on the previously defined probability space: for $t \in [0,T]$, $\omega \in \Omega$,
\begin{equation}
\label{eq: x0}
dX_i(t) = X_i(t)(\beta_i(t, \omega) dt + \sum_{k=1}^n \sigma_{ik}(t, \omega) dB_k(t)),\ \  i= 1, \ldots, n,
\end{equation}
 with initial condition $X_i(0) = x_i$. In this paper, we assume that $\text{dim}(B(\cdot)) = \text{dim}(\mathcal{X}(\cdot)) = n$, that is, we have exactly as many sources of randomness as there are stocks on the market $\mathcal{M}$. Here, $^\prime$ stands for the transpose of matrices. $\beta(\cdot) = (\beta_1(\cdot), \ldots, \beta_n(\cdot))' : [0, T] \times \Omega \rightarrow \mathbb{R}^n$ as the mean rates of return for $n$ stocks and $\sigma(\cdot) = (\sigma_{ik}(\cdot))_{n \times n} : [0, T] \times \Omega \rightarrow \text{GL}(n)$ as volatilities are assumed to be invertible, $\mathbb{F}$-progressively measurable in which $\text{GL}(n)$ is the space of $n \times n$ invertible real matrices. Then $B(\cdot)$ is adapted to $\mathbb{P}$-augmentation of the filtration $\mathbb{F}$. To satisfy the integrability condition, we assume
\begin{equation}
\label{xcond}
 \sum_{i=1}^n \int_0^T \bigg (|\beta_i(t,\omega)|+ \alpha_{ii}(t,\omega) \bigg)dt < \infty, 
\end{equation}
where $\alpha(\cdot) := \sigma(\cdot)\sigma'(\cdot)$, and its $i,j$ element $\alpha_{i,j}(\cdot)$ is the covariance process between the logarithms of $X_i$ and $X_{j}$ for $1\le i,j \le n$. The market $\mathcal{M}$ is therefore a complete market. 

\subsection{Investors}
\label{sec:investors}

Next, we define the admissible strategies used by the investors and their wealth. 
The filtered probability space we consider 
should support Brownian motion $B$ and {a random vector of type $(v_0, c)$ of a generic investor, independent of $B$, by extension of the probability space. Here, investor $\ell$'s type is determined at time $0$ as a realization $(v_0^{\ell}, c_{\ell})$ of the random vector, representing the initial wealth and personal preference of this investor. We explain the exact definition of $c$ in relative arbitrage in Section~\ref{sec: ra}.}
More specifically, we define the filtration $(\mathcal{F}^{MF}_t)_{t\in[0,T]}$ that represents the ``flow of information'' in the market and investors, where $\mathcal{F}^{MF} := \mathcal{F}^B \vee \sigma(\{v_0, c\})$.

\begin{definition}[Investment strategy]
\label{portfoliopi}
\begin{enumerate}[label={(\arabic*)}]
\item We call an $\mathbb{F}^{MF}$-progressively measurable $n$-dimensional process $\pi^\ell (\cdot)= (\pi^{\ell}_1(\cdot) , \ldots, \pi^{\ell}_n(\cdot) )'$ 
an {admissible} investment strategy if
\begin{equation}
\label{admpi}
\int_0^T (| (\pi^{\ell}(t, \omega))^{\prime}\beta(t, \omega)| + (\pi^{\ell}(t, \omega))^{\prime}\alpha(t, \omega)\pi^{\ell}(t, \omega) ) dt < \infty, \quad  \,  T\in (0,\infty),\,\, \text{a.e. } \, \omega \in \Omega.
\end{equation}
The strategy here is a self-financing portfolio, since wealth at any time is obtained by trading the initial wealth according to the strategy $\pi(\cdot)$.
We denote the admissible set by $\mathbb{A}$. In the remainder of the paper, we only consider strategies in the admissible set $\mathbb{A}$.

\item  
An investor $\ell$ uses the proportion $\pi_i^{\ell}(\cdot)$ of current wealth $V^{\ell}(\cdot)$ to invest in the stock $i$. The proportion $ \pi_0^{\ell}(\cdot) = 1-\sum_{i=1}^n \pi^{\ell}_i(\cdot)$ is on the money market for $\ell = 1, \ldots, N$. {In this paper, we consider admissible strategies that are bounded portfolios, i.e., $\pi^\ell$ takes values in a nonempty, compact, convex subset
$$K \subsetneq \{\pi = (\pi_1,...,\pi_n)\in \mathbb{R}^n \,| \pi_1 + \ldots +\pi_n = 1\}.$$
A special example is a long-only portfolio which takes values in the set
$$\Delta_n := \{\pi = (\pi_1,...,\pi_n)\in \mathbb{R}^n \,|\, \pi_1\geq 0, \ldots , \pi_n \geq 0\,; \,\pi_1 + \ldots +\pi_n = 1\}.$$}
\end{enumerate}
The dynamics of the wealth process $V^{\ell}(\cdot)$ of an individual investor $\ell$, invested in the stock market, is determined by 
\begin{equation}
\label{wealth}
    \frac{dV^{\ell}(t)}{V^{\ell}(t)} = \sum_{i=1}^n \pi_i^{\ell}(t) \frac{dX_i(t)}{X_i(t)}, \quad  V^{\ell}(0) = v^{\ell}_{0}.
\end{equation}
Consequently, 
\begin{equation}
     V^\ell(t) = v \exp \left\{ \int_0^{t} {\pi^{\ell}}'(s, \omega) \left(\beta(s, \omega) - \frac{\alpha(s, \omega)}{2} \right)ds + \int_0^t {\pi^{\ell}}'(s, \omega) \sigma(s, \omega) dB(s) \right\}. 
\end{equation}
\end{definition}
{The noise $B$ introduced in market dynamics is shared among investors. To align with the convention in mean field games literature, we refer to $B$ as the common noise, and we will construct the market model and optimization problem in later sections for a generic/representative player. Thus, we denote the wealth, the strategy, and the preference of this player as $V$, $\pi$, and $c$, respectively, without a superscript $\ell$. This representative player viewpoint will be justified in Section~\ref{sec: mfg}.}

One particular example is the {\it market portfolio} $\mathbf{m}$,
\begin{equation}
\label{Mportfolio}
\mathbf{m}_i(t) = \pi^{\mathbf{m}}_i(t) := \frac{X_i(t)}{X(t)}, \quad i=1, \ldots, n, \quad t \ge 0,
\end{equation}
where
\begin{equation} \label{Mportfolio-2}
X(t) := X_1(t) + \ldots + X_n(t), \quad t \in (0,T]; \quad X(0) := x_0 := x_1 + \cdots + x_n
\end{equation}
is the capitalization of the entire market.
Investing with a market portfolio amounts to the ownership of the entire market by investing in proportion to the market weight of each stock, as
\begin{equation}
\label{marketweal}
    \frac{d V^{\mathbf{m}}(t)}{V^{\mathbf{m}}(t)} = \sum_{i=1}^n \pi^{\mathbf{m}}_i(t) \cdot \frac{dX_i(t)}{X_i(t)} = \frac{d X(t)}{X(t)} , \quad t \ge 0. 
\end{equation}
We left the input of the market coefficients, $\omega$, in generic form for simplification. In the following section, we specify the market system adopted throughout this paper.

\subsection{Interacting market system}
\label{sec: marketz}

{Based on the supply and demand relationship for stock shares and equity capitalization, and the idea of factor models in the Capital Asset Pricing Models, we construct the capitalization process to depend on the average capital invested $\mathcal Z(t) := (\mathcal Z_{1}(t), \ldots, \mathcal Z_{n}(t) )^\prime $, $t \ge 0 $, as elaborated in the definition below. 
The factor $\mathcal Z(t)$ is an interaction term among investors. The market coefficients of the system $\{\mathcal{X}, \mathcal{Z}\}$ are thus assumed to be  \textit{time homogeneous} in the rest of the paper.} {Wealth and strategy arguments written without a superscript $\ell$ are those of a representative player.} {Note that an implicit assumption here is that we can use the expectation of the wealth (corresp. trading volume) of a representative player to represent the average wealth (corresp. trading volume) of the infinite number of investors. This setup will be explained in Section~\ref{sec: ra} and will be formally stated in Section~\ref{lims}.}
 

\begin{definition}[Market system]
\label{def:eqz}
For a given control process $\pi \in \mathbb A$, we denote the \textbf{interaction process} as $\mathcal{Z}_i(t) := \E[V(t) \pi_i(t) |\mathcal F^{B}_t]$, {which is the expectation of a representative player's trading volume.} Define $\mathcal Z(t) := (\mathcal{Z}_1(t), \ldots, \mathcal{Z}_n(t))$. Then, the market system we consider in this paper is a pair $\left(\mathcal{X}(t),\mathcal Z(t) \right) \in \mathbb R^n_+ \times \mathbb R_+$, for $t \in [0,T]$, that satisfies for $i= 1, \ldots, n,$
\begin{equation}
\label{eq: x}
dX_i(t) = X_i(t)(\beta_i(\mathcal{X}(t), \mathcal{Z}(t)) dt + \sum_{k=1}^n \sigma_{ik}(\mathcal{X}(t), \mathcal{Z}(t)) dB_k(t)),
\end{equation}
\begin{equation}
\label{eq: z_c}
    d \mathcal{Z}_i(t) = d \E[V(t) \pi_i(t) |\mathcal F^{B}_t] = \gamma_i(\mathcal{X}(t), \mathcal{Z}(t)) dt + \sum_{k=1}^n \tau_{ik}(\mathcal{X}(t), \mathcal{Z}(t)) dB_k(t),
\end{equation}

\end{definition}

{Note that we do not specify the form of $\gamma$ and $\tau$ in the above, as this would be determined from a fixed point argument specified in Section~\ref{sec: mfg}, through Mean Field Equilibrium. The strategies at the equilibrium will specify the dynamics defined above.}

Since $\sigma(\cdot)$ is invertible, there exists an $\mathbb{F}$-progressively measurable process $\theta(\cdot)$ such that for any $(\mathcal{X}(t), \mathcal{Z}(t)) \in \mathbb{R}_+^n \times \mathbb{R}_+^n$, $t \in (0, \infty)$,
$\sigma(\mathcal{X}(t), \mathcal{Z}(t)) \theta(\mathcal{X}(t), \mathcal{Z}(t)) = \beta(\mathcal{X}(t), \mathcal{Z}(t)),$
$\tau(\mathcal{X}(t), \mathcal{Z}(t)) \theta(\mathcal{X}(t), \mathcal{Z}(t)) = \gamma(\mathcal{X}(t), \mathcal{Z}(t))$, and
\[
\mathbb{P} \bigg( \int_0^T |\theta(\mathcal{X}(t), \mathcal{Z}(t))|^2 dt < \infty, \forall T \in (0, \infty) \bigg) = 1.
\]
In the scope of a complete market, the above shows that the price of risk process $\theta(t)$ governs both the risk premium per unit volatility of stocks and the trading volumes, since the market is simultaneously defined by the stocks and investors.
Now we can define the deflator $L(t)$ based on the market price of the risk process,
{
\begin{equation}
    \label{ltmg}
dL(t) = - \theta(t) L(t) dB(t), \, t > 0; \quad L(0) = 1.
\end{equation}
}
Note that, under the existence of a market price of risk process, the market is endowed with the existence of a \textit{local} martingale $L(\cdot)$ with $\E[L(T)] \leq 1$. {As opposed to the typical risk-neutral measure construction, we do not enforce the martingale assumption $\E[L(T)] = 1$ of such deflator (which can be guaranteed by Novikov's condition \cite{ikeda2014stochastic}, $\E\left[\exp \Big(\frac{1}{2} \int_0^T |\theta(t)|^2 dt \Big)  \right]<\infty$). When $L(\cdot)$ is a martingale, the market excludes arbitrage opportunities.}

We give the following Assumption~\ref{xv}  for the market capitalization, wealth and preference of investors.
\begin{asmp}
\label{xv}
We assume the following:
\begin{enumerate}[label=(\arabic*)]

\item \label{thomo}
The market coefficients $\beta(\cdot)$, $\sigma(\cdot)$, $\gamma(\cdot)$ and $\tau(\cdot)$ in \eqref{eq: x}-\eqref{eq: z_c} take values in $\mathbb{R}_+^n \times \mathbb{R}_+^n$, are time-homogeneous and the process $(\mathcal X(t), \mathcal{Z}(t)), t \ge 0  $ in Definition \ref{portfoliopi} is Markovian, i.e.,
\[
X_i(t) \beta_i(\mathcal{X}(t), \mathcal{Z}(t)) = b_i(\mathcal{X}(t), \mathcal{Z}(t)), \quad X_i(t) \sigma_{ik} (\mathcal{X}(t), \mathcal{Z}(t)) = s_{ik}(\mathcal{X}(t), \mathcal{Z}(t)),\] 
\[
\sum_{k=1}^n s_{ik}(\mathcal{X}(t), \mathcal{Z}(t))s_{jk}(\mathcal{X}(t), \mathcal{Z}(t)) = a_{ij}(\mathcal{X}(t), \mathcal{Z}(t)). 
\]

\item The Lipschitz continuity and linear growth condition are satisfied with Borel measurable mappings $b_{i}(x, z)$, $s_{ik}(x, z)$ in \eqref{eq: x}-\eqref{eq: z_c} from $ \mathbb{R}^n_+ \times \mathbb{R}^n_+$ to $\mathbb{R}^n$. That is, there exist some constants $C_L \in (0, \infty)$ and $C_G \in (0, \infty)$, such that
\begin{equation}
\label{bslip}
    |b(x, z) - b(\widetilde{x}, \widetilde{z})| + |s(x, z) - s(\widetilde{x}, \widetilde{z})| \leq C_L[|x - \widetilde{x}| +|z - \widetilde{z}| ] , 
\end{equation}
\[
|b(x, z)| + |s(x, z)| \leq C_G(1+|x|+|z|).
\]
for every $z, \widetilde{z} \in \mathbb{R}_+$, 
$x, \widetilde{x} \in \mathbb R_{+}^n$. We assume these conditions for $\gamma_i(x,z)$ and $\tau_{ik}(x,z)$ as well.
\end{enumerate}
\end{asmp}

\begin{prop}
\label{eumckean-1}
Suppose that the initial capitalization and trading volume $(\mathbf{x}, \mathbf{z})$ are independent of the Brownian motion $B (\cdot)$ on $(\Omega, \mathcal F, \mathbb F, \mathbb P)$, and they satisfy $\E[\|(\mathbf{x}, \mathbf{z})\|^2] \leq \infty$. Then, under Assumption~\ref{xv}, the McKean-Vlasov system \eqref{eq: x}-\eqref{eq: z_c} admits a unique strong solution.
\end{prop}
We show a more general result in Theorem~\ref{eumckean}.

\section{Relative arbitrage opportunities for competitive investors}
\label{sec: ra}

{
The performance of a portfolio is typically measured with respect to some specific benchmark. Given two investment strategies $\pi(\cdot)$ and $\rho(\cdot)$, with the same initial capital $V^{\pi}(0) = V^{\rho}(0) =1$, we shall say that $\pi(\cdot)$ represents an arbitrage opportunity relative to $\rho(\cdot)$ over the time horizon $[0,T]$, with a given $T>0$, if
$$
\mathbb{P} \left( V^{\pi}(T) \geq V^{\rho}(T) \right) = 1 \quad \text{and} \quad \mathbb{P} \left(  V^{\pi}(T) > V^{\rho}(T) \right) >0. 
$$
Previous works in relative arbitrage have been focused on using market capitalization as a benchmark, that is, $\rho(\cdot) := \mathbf{m}(\cdot)$.  
However, when interactions among investors are considered, the market portfolio is not necessarily the only criterion of investors' performances. 
For example, asset managers care about not only absolute performance compared to the market index but also relative performance with respect to all collegiate managers; they try to exploit strategies that achieve an arbitrage relative to market and peer investors.
}

\subsection{Relative arbitrage benchmark: competing with the market and the continuum of players} 
\label{sec: benchmark}

A player competes with the market and the entire group with respect to the benchmark
\begin{equation}
\label{mathcalbench}
    \mathcal{V}(T) := \delta \cdot X(T) + (1-\delta) \cdot \E[V_T | \mathcal{F}_T^B],
\end{equation}
is the weighted sum of total capitalization and the conditional expectation of wealth given $\mathcal F_\cdot^B$. The formal definition of the conditional measure will be given in Definition \ref{def: mfe}. 
To better understand the market and the interactions, we first prove the following result, which shows that the benchmark $\mathcal{V}$ is a valid wealth process.
\begin{prop}
\label{pistarr}
Benchmark $\mathcal{V}(t)$ in \eqref{mathcalbench} can be generated from a strategy $\Pi(\cdot) := (\Pi_1(\cdot), \ldots, \Pi_n(\cdot)) \in \mathbb{A}$ that satisfies 
\begin{equation}
\label{bigpi}
\Pi_i(t, \mathcal X(t), \mathcal{Z}(t)) = \frac{1}{\mathcal{V}(t)} \left( \delta X_i(t) + (1-\delta) \mathcal{Z}_i(t)\right). 
\end{equation}
\end{prop}

We assume that each investor $\ell$ measures the logarithmic ratio of their own wealth $V^\ell(T)$ at time $T$ to the relative arbitrage benchmark \eqref{mathcalbench}, and searches for a strategy with which the logarithmic ratio is above a personal level of preference almost surely. {We denote the investment preferences of investors by $c_{\ell}, \ell \in \mathbb N$, i.i.d samples drawn from the distribution $Law(c)$ and assigned to investor $\ell$ at time 0. $Law(c)$ is independent of common noise $B$}. An investor $\ell$ tries to achieve the relative amount $e^{c_{\ell}}\mathcal{V}(T)$ to the benchmark $\mathcal V(T)$  based on their preferences $c_\ell$, that is,   
\begin{equation} 
\label{eq: arbitrage3.2}
\log \frac{V^{\ell}(T)}{\mathcal{V}(T)} \ge c_{\ell}, \quad \text{a.s.} \quad 
\text{ or equivalently, } \quad 
V^{\ell}(T) \geq e^{c_{\ell}} \mathcal{V}(T), \quad \text{a.s.}.
\end{equation}

\begin{remark}
{If $c_\ell$ above is taken to be $0$ for every $\ell$, then it means that the investor tries to match or exceed the exact benchmark quantity. In the regime with finite number of investors \cite{ichiba2020relative}, it has been shown that a necessary condition for investors to achieve relative arbitrage is $\frac{1-\delta}{N} \sum_{\ell = 1}^N e^{c_\ell} < 1$ and it is possible for $c_\ell$ to take a small positive number. To see this, consider a special case that the preferences held by investors are the same, then $c_\ell =c < \log \left(\frac{1}{1-\delta}\right)$ is a valid level of preference where $\delta \in (0,1]$. }

{In the setup of this paper, the values $c_\ell$ are modeled as random samples from a common distribution. The market system is coupled among investors through the strategies, trading volume, and objective \eqref{utobj1}, which makes it difficult to optimize the objective analytically because of the interactions, as is the case in \cite{ichiba2020relative}. Thus, it is advantageous to consider each player $\ell \in \mathbb{N}$ is conditionally i.i.d given common noise $B$ in the infinite-player case. In this case, any player with an index $\ell$ is indistinguishable from other players, and the system is decoupled as the interaction terms are now constructed by a population distribution rather then specific pairwise interactions for each $\ell$. In the rest of this section, we consider this simplification and leave out the superscript $\ell$ in all related quantities. The justification of such representative player construction will be deferred to Section~\ref{lims}.}
\end{remark}

\subsection{Optimization of relative arbitrage opportunities}
\label{sec: u}

{
A natural question in pursuing relative arbitrage opportunities is how to formulate and solve the optimization problem that maximizes the (expected) return. In an interacting system, this is a challenging task, as the decisions feed back into both the market dynamics and the objective functional. In this section, we specify the optimization problem in pursuing arbitrage opportunities relative to the benchmark we defined in Section~\ref{sec: benchmark}. {A subject in a similar favor as ours is quantile hedging \cite{follmer1999quantile}, where the investors might not be willing to perform a `perfect' hedging and the question becomes what the maximal probability of `success' the investor can achieve with a given smaller amount of initial capital.}
}

\begin{definition}[Optimal arbitrage]
\label{def: ut}
For $\ell \in \mathbb N$, and initial time $t \in [0,T)$, we define the optimal arbitrage $u^{\ell}: (0,\infty) \times (0,\infty)^n \times (0,\infty)^n \rightarrow (0,\infty)$ as the optimal fraction of initial wealth of the benchmark portfolio $\widetilde{\mathcal{V}}(t)$. With the initial wealth $u^{\ell}(T-t, \mathbf{x}, \mathbf{z}) \widetilde{\mathcal{V}}(t)$, we match or exceed the benchmark portfolio at the terminal time $T$. {That is, given the initial condition $(\mathcal{X}(t),\mathcal{Z}(t)) := ( \mathbf{x}, \mathbf{z}) \in (0,\infty)^n \times (0,\infty)^n$, the initial wealth of the other investors $v^{-\ell}:= (v^{k}(\cdot))_{k \in \mathbb{N}, k \neq \ell}$, and the admissible portfolios $\pi^{-\ell}(\cdot):= (\pi^{k}(\cdot))_{k \in \mathbb{N}, k \neq \ell}$.}
\begin{equation}
\label{utobj1}
\begin{aligned}
    u^{\ell}(T-t) = \inf \bigg \{ \omega^{\ell} \in (0, \infty) \,  \Big \vert \, \exists
\widetilde{\pi}^{\ell}(\cdot) \in \mathbb{A} & \text{ such that } \\
& \,v^{\ell} = \omega^{\ell} \widetilde{\mathcal{V}}(t) , \, \,  \widetilde{V}^{{\ell}}(T) \geq e^{c_{\ell}} \cdot \widetilde{\mathcal{V}}(T) \bigg \} 
\end{aligned}
\end{equation}
for $ 0 \le t \le T$ and $\ell \in \mathbb N$. Here, $\widetilde{\mathcal V}$ and $\widetilde{V}^\ell$ are the benchmark and portfolio wealth corresponding to $\tilde{\pi}$'s, respectively. 
At every time $t$, each investor 
optimizes $\, \widetilde{\pi}^\ell(\cdot)\, $ from $t$ to $T$, in order to get the optimal quantity as defined in \eqref{utobj1}. In other words, here $\widetilde{V}^{v^\ell , \pi^\ell} (T) \,$ is generated by the admissible portfolio $\, \{ \pi^\ell (s) , s \ge t\} \, $ starting from time $t \ge 0$. We consider $\, \widetilde{\pi}^\ell(s)\, $, $\, t \le s \le T$ such that the corresponding portfolio value $\, \widetilde{V}^\ell (s)\, $, $\, t \le s \le T\, $ satisfies $\, \widetilde{V}^\ell (T) \ge e^{c_\ell} \cdot \widetilde{\mathcal{V}}(T)$, where $\widetilde{V}^{\ell} (t) = v^\ell = \omega^\ell \widetilde{\mathcal{V}}(t) \, ,$
\[
\frac{d \widetilde{V}^{\ell}(s)}{ \widetilde{V}^{\ell} (s)} = \sum_{i=1}^n \widetilde{\pi}_i^\ell (s) \cdot \frac{d X_i (s)}{\, X_i (s) \,}\, ; \quad t \le s \le T, \, \, \ell \in \mathbb N.
\]
\end{definition}
{We use the notation $\widetilde{\mathcal{V}}(\cdot)$ instead of the initially defined ${\mathcal{V}}(\cdot)$ to emphasize that each subproblem starting from $t$ needs to be solved separately for the time horizon $[t,T]$. At the initial time $t$, $\mathbf{z}$ is not known a priori, and Definition~\ref{def: ut} results in a fixed point problem on the path space of $u$, taking values from $C([0,T], \mathbb R_+)$, then $\mathbf{z}$ is solved at the initial time. This will be clarified and derived in Section~\ref{sec: mfg}.}

Next, we present the probabilistic characterization of the optimal arbitrage. 
\begin{prop}
\label{f1}
Under Assumption~\ref{xv}, 
{assume $\pi(s)$ to be the optimal strategy over $s \in [t,T]$ in Definition~\ref{def: ut} that achieves $u(T-t, \mathbf{x}, \mathbf{z})\in C^{1,3,3}$}. 
With given portfolio $\pi^{-\ell}(\cdot)$, for $\ell = 1, \ldots, N$, and the initial values $(\mathbf{x}, \mathbf{z})$, 
$u(T)$ in \eqref{utobj1} can be derived as $e^{c} \mathcal{V}(T)$'s discounted expected values 
\begin{equation}
\label{ute}
        u(T, \mathbf{x}, \mathbf{z}) = e^{c} \mathbb{E} \big[ \mathcal{V}(T) L(T) | \mathcal{F}_0^{MF}\big]\, /\, \mathcal{V}(0). 
\end{equation}
Moreover, the Markovian property of $u(\cdot)$ gives
\begin{equation}
\label{ggg}
   u(T-t, \mathcal{X}(t), \mathcal{Z}(t)) = e^{c} \frac{\E[\mathcal{V}(T)L(T) | \mathcal{F}^{MF}_t] }{\mathcal{V}(t)L(t)}. 
\end{equation}
\end{prop}

The proof of this is based on the construction of strategies that generate the right hand side of \eqref{ggg} through the martingale representation theorem. { The choice of $u(t, \mathbf{x}, \mathbf{z}) \in C^{1,3,3}((0, \infty) \times (0, \infty)^n \times (0, \infty)^n)$ is to enforce sufficient regularity so that the optimal strategy $\pi(\cdot)$ is in the admissible set $\mathbb{A}$. We refer the readers to a similar treatment of this regularity and the finite-player counterpart of the proof in \cite[Section 3]{ichiba2020relative}. }
From this probabilistic formulation, we can further characterize it as the nonnegative minimal solution of a Cauchy problem, derived in section ~\ref{5.1}.
{
Given the existence of relative arbitrage in the sense of \eqref{eq: arbitrage3.2},  the process $L(\cdot)$ is a strict local martingale, i.e., $\mathbb{E}[L(T)]<1$. This is explained in \cite{ichiba2020relative} for a finite-investor interacting market system and it can be shown in the same vein here. Readers can refer to \cite[Remark 1]{ichiba2020relative} and \cite[Proposition 6.1]{karatzas2009stochastic} for details. The discounted processes $\widehat{V}(\cdot) := V(\cdot) L(\cdot)$, and $\widehat{X}(\cdot) := X(\cdot) L(\cdot)$ satisfy the following equations,  
\begin{equation}
\label{vlhatt}
\begin{aligned}
   d \widehat{V}(t) = d V(t) L(t) =  \widehat{V}(t)\big(\pi^{\prime}(t) \sigma(t) -  \theta'(t) \big)dB(t); \quad  \widehat{V}(0) = \widehat{v}, \, \, \ell = 1, \ldots , N\, , 
\end{aligned}
\end{equation}
\begin{equation}
\label{xhat}
d \widehat{X}(t) = \widehat{X}(t) \sum_{k=1}^n (\sum_{i=1}^n \mathbf{m}_i(t) \sigma_{ik}(t) - \theta_k(t)) dB_k(t); \quad \widehat{X}(0) = x.
\end{equation}
}

\begin{remark}[$\mathcal{F}^{B}$ in the market dynamics]
Note that $\mathbf{z}$ needs to be determined with the knowledge of $u(T, \mathbf{x}, \mathbf{z})$. Thus, we can construct a fixed point problem on the path space of $u$, taking values from $C([0,T], \mathbb R_+)$, then $\mathbf{z}$ is solved at the initial time. We discuss this in detail in the next section. With $\mathbf{z}$ fixed, the market $(\mathcal X(t), \mathcal{Z}(t))$ only has access to the information of the common noise, and is $\mathcal{F}_t^B$-measurable. When the trajectory $(B_t)_{t \in [0,T]}$ is given, the randomness of the conditional expectation in $\mathcal{Z}$
comes from the preference parameter $c$ and the initial trading volume $\mathbf{z}$. 
\end{remark}
\begin{remark}[$\mathcal{F}^{MF}$ in the optimization]
The information available to each player $\ell$ at time $t$ is the common noise and their private state of the type vector $(v_0^\ell, c_\ell)$. With a given initial condition $(\mathbf{x}, \mathbf{z}) := \left(\{X_i(0)\}_{i=1}^n, \{\mathcal Z_i(0)\}_{i=1}^n\right)$, we search for the optimal arbitrage and the corresponding $\mathcal{F}^{MF}$-measurable strategy $\pi_t$, for $t \in [0,T]$. Thus, the admissible strategies are progressively measurable with respect to $\mathcal{F}_t^{MF}$. This means that the randomness of wealth processes is from $\mathcal{F}^B$, $\pi(\cdot)$ and the path $(\mathbf{x}, \mathbf{z})$ over the time horizon. $\pi(\cdot)$ depends on $\{u(T-t, \mathbf{x}, \mathbf{z})\}_{t \in [0,T]}$ and $(\mathcal{X},\mathcal{Z})$ and, therefore, depends on $c$.  
\end{remark}

{
\subsection{PDE characterization of the best relative arbitrage}
\label{5.1}
}

{
This section elaborates the derivation of the Cauchy problem, which is derived in the same vein as the finite-player counterpart in \cite{ichiba2020relative}.
Note that the coefficients $(\gamma, \tau)$ are not given a-priori and are solved through a fixed point problem, see \cite[Appendix D.2]{ichiba2020relative}. We consider $u(t, \mathbf{x}, \mathbf{y}) \in C^{1,3,3}((0, \infty) \times (0, \infty)^n \times (0, \infty)^n)$ which guarantees that the optimal strategy $\pi$ is in the admissible set $\mathbb{A}$ in Definition \ref{portfoliopi}.} 

{
In the following, we use the notation $D_i$ and  $D_{ij}$  for  the  partial and second partial derivative with respect to the $i$ th or $i$ th and $j$ th variables  in $\mathbf x$, respectively; $D_p$ and  $D_{pq}$  for  the first and second partial derivatives in the variable $\mathbf z$.}

{
\begin{asmp}
\label{hasmp}
There exists a function $H: \mathbb{R}_+^n \times \mathbb{R}_+^n \rightarrow \mathbb{R}_+^n$ of class $C^2$, such that
\[
b( \mathbf{x}, \mathbf{z}) = 2a( \mathbf{x}, \mathbf{z}) D_x H( \mathbf{x},\mathbf{z}), \quad \gamma( \mathbf{x}, \mathbf{z}) = 2\psi( \mathbf{x}, \mathbf{z}) D_z H(\mathbf{x}, \mathbf{z}),
\]
i.e.,  $b
_i(\cdot) =\sum_{j=1}^n a_{ij}(\cdot) D_{j} H (\cdot)$,  $\gamma
_p(\cdot) =\sum_{q=1}^n \psi_{pq}(\cdot) D_{q} H(\cdot)$ in component wise for $i,p = 1, \ldots , n$. 
\end{asmp}
{The above assumption means that the vector field $a^{-1}(\cdot) b(\cdot)$ and $\psi^{-1}(\cdot) \gamma(\cdot)$ are conservative in $\mathbf{x}$ and $\mathbf{z}$, respectively. They have a symmetry relationship that $D_z\left( a^{-1}(\cdot) b(\cdot) \right) = D_x \left(\psi^{-1}(\cdot) \gamma(\cdot) \right)$. This assumption is originally posed in \cite{fernholz2010optimal} for the case without mean-field structure of the market. For a quick illustrative example, consider 
\[
b_i(\mathbf{x},\mathbf{z})= 1 + \frac{1}{2} x_i \sum_{i=1}^n z_i , \quad a_{ij}(\mathbf{x}, \mathbf{z}) = x_i \delta_{ij},
\]
\[
\gamma_i(\mathbf{x},\mathbf{z})= 1 + \frac{1}{2} z_i \sum_{i=1}^n x_i , \quad
\psi_{ij}(\mathbf{x}, \mathbf{z}) = z_i \delta_{ij},
\]
where $\delta_{ij} = 1$, when $i = j$; and $\delta_{ij} = 0$ otherwise, when $i \neq j$. Then,
\[
H(\mathbf{x}, \mathbf{z}) = \sum_{i=1}^n (\log x_i + \log z_i) + \frac{1}{2} \left(\sum_{i=1}^n x_i\right) \left(\sum_{i=1}^n z_i\right).
\]
In financial markets, the so-called \emph{leverage effect} suggests that the smaller stocks tend to have greater volatility than the larger stocks. Similarly, lower trading volumes in a stock are often associated with higher volatility in trading. Later in Example~\ref{ex: vsm}, we discuss a realistic volatility-stabilized market model that satisfies this effect, and provide a streamlined explanation of Assumption~\ref{hasmp} and the corresponding optimal arbitrage solution. 
}
Use the definition of $\theta(\cdot)$ and Assumption~\ref{hasmp},
 and apply It\^o's lemma on $H(\cdot)$, it follows that
\begin{equation}
\label{ltmfg}
    \begin{aligned}
    L(t) = \exp \bigg\{ - H(\mathcal{X}(t), \mathcal{Z}(t)) +  H(\mathbf{x}, \mathbf{z}) - \int_0^{t} ( k(\mathcal{X}(s), \mathcal{Z}(s)) + \Tilde{k}(\mathcal{X}(s), \mathcal{Z}(s)) ) ds \bigg\}
\end{aligned}
\end{equation}
for $t \ge 0 $, where
\[
k(\mathbf{x}, \mathbf{z}) := - \sum_{i=1}^n \sum_{j=1}^n \frac{a_{ij}( \mathbf{x}, \mathbf{z})}{2} [ D_{ij}^2 H( \mathbf{x},\mathbf{z}) + 3 D_i H ( \mathbf{x},\mathbf{z}) D_j H( \mathbf{x},\mathbf{z})],
\]
\[ 
\begin{aligned}
\Tilde{k}( \mathbf{x}, \mathbf{z}) :&= -\sum_{i=1}^n \sum_{j=1}^n \frac{\psi_{pq}( \mathbf{x}, \mathbf{z})}{2} [ D_{pq}^2 H(\mathbf{x},\mathbf{z}) + 3 D_p H(\mathbf{x},\mathbf{z}) D_q H(\mathbf{x},\mathbf{z})] \\
& \qquad {}+ \sum_{i=1}^n \sum_{p=1}^n (s \tau^\top)_{ip} D_i H ( \mathbf{x},\mathbf{z}) D_p H( \mathbf{x},\mathbf{z}) 
\end{aligned}
\]
for $(\mathbf{x},\mathbf{z}) \in (0,\infty)^n \times (0,\infty)^n$. 
Thus, we can rewrite $u(\cdot)$ in Proposition \ref{f1} and get
\begin{equation}
\label{gg}
   u(\tau, \mathbf{x}, \mathbf{z}) =e^{c}  \frac{G(\tau,\mathbf{x}, \mathbf{z})}{g(\mathbf{x}, \mathbf{z})},
\end{equation}
where 
\begin{equation*}
\begin{split}
g(\mathbf{x}, \mathbf{z}) &:= \bigg(\delta \sum_{i=1}^n x_i + (1-\delta) \sum_{i=1}^n z_i \bigg) e^{-H(\mathbf{x},\mathbf{z})}, 
\\
G(T,\mathbf{x}, \mathbf{z}) &:=  \E^{\mathbf{x}, \mathbf{z}} \left[ g(\mathcal{X}(T), \mathcal{Z}(T)) e^{- \int_0^{T} k(\mathcal{X}(t),\mathcal{Z}(t))+ \Tilde{k}(\mathcal{X}(t),\mathcal{Z}(t)) dt}\right].
\end{split}
\end{equation*}
\begin{asmp}
\label{hh}
Assume that $g(\cdot)$ is H\"older continuous, uniformly on compact subsets of $\mathbb{R}_+^n \times \mathbb{R}_+^n$. $G(\cdot)$ is continuous on $(0, \infty) \times (0, \infty)^n \times (0, \infty)^n$, of class $C^{1,3,3}((0, \infty) \times (0, \infty)^n \times (0, \infty))$.
The function $G(\cdot)$ yields the following dynamics by Feynman-Kac formula,
\begin{equation}
\begin{aligned}
    \frac{\partial G}{\partial \tau}(\tau, \mathbf{x}, \mathbf{z}) = & \, \mathcal{L} G(\tau,  \mathbf{x}, \mathbf{z}) - (k(\mathbf{x}, \mathbf{z})+ \Tilde{k}(\mathbf{x}, \mathbf{z}) ) G(\tau,  \mathbf{x}, \mathbf{z}), \\ 
G(0, \mathbf{x}, \mathbf{z}) = & \, g( \mathbf{x}, \mathbf{z}),\  
\end{aligned}
\end{equation}
for $(\tau,  \mathbf{x}, \mathbf{z}) \in \mathbb{R}_+ \times \mathbb{R}_+^n \times \mathbb{R}_+^n$. 
{The generator for the process $(\mathcal X(\cdot), \mathcal Z(\cdot)) $ can be written as
\[
\begin{aligned}
 \mathcal{L}f & := \sum_{i=1}^n \sum_{j=1}^n a_{ij}( \mathbf{x}, \mathbf{z}) \big[ \frac{1}{2}  D_{ij} f + 2 D_i f D_j H( \mathbf{x}, \mathbf{z}) \big] \\
 & + \sum_{p=1}^n \sum_{q=1}^n \psi_{pq}( \mathbf{x}, \mathbf{z}) \big[ \frac{1}{2}  D_{pq} f + 2 D_p f D_q H( \mathbf{x}, \mathbf{z}) \big]\\
 &+ \frac{1}{2} \sum_{i=1}^n \sum_{p=1}^n (s \tau')_{ip}( \mathbf{x}, \mathbf{z}) D_{ip} f + \frac{1}{2} \sum_{i=1}^n \sum_{p=1}^n (\tau s')_{pi}( \mathbf{x}, \mathbf{z}) D_{pi} f,
 \end{aligned}
\]
where $(\tau s')_{pi}( \mathbf{x}, \mathbf{z}) = (s \tau')_{ip}( \mathbf{x}, \mathbf{z}) = \sum_{k=1}^K s_{ik}(\mathbf{x}, \mathbf{z}) \tau_{pk}( \mathbf{x}, \mathbf{z})$.}
\end{asmp}

Under Assumption~\ref{hh}, $u(\tau, \mathbf{x}, \mathbf{z}) \in C^{1,3,3}((0, \infty) \times (0, \infty)^n \times (0, \infty))$ is bounded on $K \times (0, \infty)^n \times (0, \infty)^n$ for each compact $K \subset (0, \infty)$. Plugging \eqref{ggg} into the set of above equations and using the Markovian property of $g(\cdot)$ gives
\[
\begin{aligned}
 \frac{\partial u(t,  \mathbf{x}, \mathbf{z})}{\partial t} g( \mathbf{x}, \mathbf{z}) &= \mathcal{L}(u(t,  \mathbf{x}, \mathbf{z}) g( \mathbf{x}, \mathbf{z})) - \left( k( \mathbf{x}, \mathbf{z}) + \Tilde{k}( \mathbf{x}, \mathbf{z})\right) u(t,  \mathbf{x}, \mathbf{z}) g( \mathbf{x}, \mathbf{z}).
\end{aligned}
\] 
For simplicity, we write $u(t)$ in place of $u(t, \mathbf{x},\mathbf{z})$. 
}

{
We can show that \eqref{ggg} solves the following Cauchy problem in the same vein as the finite player case in \cite{ichiba2020relative}. However, note that \eqref{ggg} solves a single Cauchy problem as opposed to the coupled $N$-dimensional PDEs system in $N$-player game from \cite{ichiba2020relative}. }

{
\begin{prop} 
Under Assumption~\ref{xv}-\ref{hh}, the function $u: [0,\infty) \times (0,\infty)^n \times (0,\infty) \rightarrow (0, 1]$, 
is the smallest non-negative continuous function, of class $C^{1,3,3}$ on $(0,\infty) \times (0,\infty)^n \times (0, \infty)$, that satisfies 
\begin{equation}
\label{inequ}
\begin{aligned}
\frac{\partial u(\tau, \mathbf{x}, \mathbf{z})}{\partial \tau} \geq \mathcal{A} u(\tau, \mathbf{x}, \mathbf{z}), \quad u(0, \mathbf{x}, \mathbf{z}) = e^c,
 \end{aligned}
\end{equation}
where 
\begin{equation}
\label{eqsol}
\begin{aligned}
\mathcal{A} u(\tau, \mathbf{x}, \mathbf{z}) =& \frac{1}{2} \sum_{i=1}^n \sum_{j=1}^n a_{ij} ( \mathbf{x}, \mathbf{z}) \left(D_{ij}^2 u(\tau,  \mathbf{x}, \mathbf{z}) + \frac{ 2 \delta D_i u(\tau,  \mathbf{x}, \mathbf{z})}{ \delta \mathbf{x} \cdot \mathbf{1} + (1-\delta) \mathbf{z} \cdot\mathbf{1}}\right) \\
& + \frac{1}{2} \sum_{p=1}^n \sum_{q=1}^n \psi_{pq} ( \mathbf{x}, \mathbf{z}) \left(D_{pq}^2 u(\tau, \mathbf{x}, \mathbf{z})+ \frac{ 2 (1-\delta) D_p u(\tau, \mathbf{x}, \mathbf{z})}{ \delta \mathbf{x} \cdot \mathbf{1} + (1-\delta) \mathbf{z} \cdot\mathbf{1}}\right)\\
& + \sum_{i=1}^n \sum_{p=1}^n (s \tau^\top)_{ip}( \mathbf{x}, \mathbf{z}) D_{ip}^2 u(\tau, \mathbf{x}, \mathbf{z})\\
& + \sum_{i=1}^n \sum_{p=1}^n (s \tau^\top)_{ip}( \mathbf{x}, \mathbf{z}) \frac{ \delta D_p u(\tau, \mathbf{x}, \mathbf{z}) + (1-\delta) D_i u(\tau, \mathbf{x}, \mathbf{z})}{ \delta \mathbf{x} \cdot \mathbf{1} + (1-\delta) \mathbf{z} \cdot\mathbf{1}}.
\end{aligned}
\end{equation}
\end{prop} 
We emphasize here that \eqref{eqsol} is determined entirely from the volatility structure of $\mathcal{X}(\cdot)$ and $\mathcal{Z}(\cdot)$. }
{
The existence of relative arbitrage opportunities amounts to the smallest non-negative solution $u(\tau, \mathbf{x}, \mathbf{z})$ of \eqref{inequ} to be less than $e^c$ for some $\tau, \mathbf{x}, \mathbf{z}$. Indeed, one can see that $u(\tau, \mathbf{x}, \mathbf{z}) \equiv 1$ is a trivial solution of the equality of \eqref{inequ}.
}

    \bigskip

    {
The evolution of market dynamics and the action of the investors are coupled, and we simplify the discussion so far by considering a given (optimal) strategy $\pi(t) \in \mathbb A$, $t \in [0,T]$. Thus, the problem is decoupled as we can determine the coefficients of the market capitalization and trading volumes \eqref{eq: z_c} under this optimal strategy. This gives the motivation of the next section about mean-field equilibrium, where we analyze how the optimization is achieved by the investors and how this in turn influence the interactions in the market.}

\section{Mean Field Games of relative arbitrage optimization}
\label{sec: mfg}

{The group of competitive investors seeks to optimize their wealth in the coupled market system of capitalization and trading volume of each stock. The market movement is affected by the group of investors and vice versa. This motivates us to formulate and solve the problem as a stochastic game. That is, we look for the fixed point of the best response function, the Markovian control $\pi^{\ell \star}$ of player $i$ that minimizes $J(\pi^{\ell}, \pi^{-\ell \star})$ over Markovian controls $\pi^{\ell}$ with given strategies $\pi^{-\ell \star}$ of the rest of the players. We have observed in \cite{ichiba2020relative} that it is unlikely to obtain a tractable equilibrium from an $N$-player game, especially when $N$ is large. Ideally, when $N$ approaches infinity, the change of strategies of an individual player has nearly no influence of the mean field interaction term and thus the interacting system can be decoupled into a system of numerous independent investors (players) that behave in a similar way. Then, we can choose one representative player to study the optimization of the entire system. }
{ In this section, we first consider the ideal case in mean field game setup to focus on the optimization of one representative player. We investigate in the next section whether mean field game is a reasonable approximation when we only have finite players in reality.}

\subsection{Mean Field Games and Equilibrium}
\label{exdmfg}

{
Every player tries to minimize the relative amount of initial capital compared to that of the benchmark $\mathcal{V}(T)$. The cost functional of an investor $\ell$ is
\[
J^\ell (\text{\boldmath$\pi$}) := \inf \bigg \{ \omega^{\ell} > 0 \  \big| \, V^{\omega^{\ell} \mathcal{V}(0), \pi^{\ell}}(T) \geq e^{c_{\ell}} \mathcal{V}(T)\bigg \},
\]
for all admissible strategy profiles $\text{\boldmath$\pi$}(\cdot) = (\pi^\ell(\cdot))_{\ell \in \mathbb{N}}$. The influence of $\pi^{k}$, $k \neq \ell$, is implicitly defined in the wealth $V(\cdot)$ and benchmark $\mathcal{V}(\cdot)$.} We simplify this notation by considering the strategy $\pi$ of a representative player:
\begin{equation}
\label{jinfeq}
\begin{aligned}
J(\pi):= \inf \bigg \{ \omega > 0 \  \big| V^{\omega \mathcal{V}(0), \pi}(T) \geq e^c \mathcal{V}(T) \bigg \},
\end{aligned}
\end{equation}
for admissible strategy $\pi(\cdot) \in \mathbb{A}$ with the initial wealth $v = \omega \mathcal{V}(0)$. The exchangeability results are shown in Section~\ref{lims}  to validate the use of the representative player in the mean field game and in the search of the optimal arbitrage.
The infimum is attained $\inf_{\pi \in \mathbb{A}} J(\pi) = u^{\star}(T)$. We provide a PDE characterization of $u^{\ell \star}$ as the minimum non-negative solution of a Cauchy equation in Appendix~\ref{5.1}. In fact, there exists $\pi \in \mathbb{A}$ such that
\[
\omega \mathcal{V}(0) \exp \bigg \{ \int_0^T \left(\pi_t\right)'(\beta_t - \frac{1}{2} \alpha_t \pi_t
)dt + \int_0^T \left(\pi_t\right)' \sigma_{i}(t) dB(t)
\bigg\} \geq e^{c} \mathcal{V}(T)
\]
and hence, regarding the infimum in \eqref{jeeq}, we have 
\begin{equation}
\label{jeeq}
    J(\pi) = \frac{\mathcal{V}(T)}{\mathcal{V}(0)} \exp \bigg\{ - \int_0^T {\left(\pi_t\right)' \Big(\beta_t - \frac{1}{2} \alpha_t \pi_t
\Big)dt} - \int_0^T \left(\pi_t\right)' \sigma (t) dB(t) \bigg\}\leq \omega.
\end{equation}

{To define the mean field equilibrium we first specify the space of the probability measures we consider in this paper. We denote by $\mathcal{P}_2(K)$ the space of probability measures on $K$ with finite second moment endowed with the Wasserstein$-2$ (\cite[Def~6.1]{villani2009optimal}) metric, $\mathcal W_2$. 
For two measures $\mu,\nu$ on a separable metric space $(K,d)$, 
\begin{equation*}
\mathcal{W}_2(\mu,\nu) 
\coloneqq \Big (\min_{\pi\in\Pi(\mu,\nu)} \int d(x,y)^2\, d\pi(x,y)\Big)^{1/2} 
= \Big (\min_{X \sim \mu, Y \sim \nu} \mathbb{E}(d(X,Y)^2)\Big)^{1/2}.
\end{equation*}}
\begin{definition}(Strong mean field equilibrium)
\label{def: mfe}
We say that the measure $\mu$ which takes value in $\mathcal{P}_2(C([0,T]; \mathbb R_+))$ is a mean field equilibrium (MFE) if there exists an admissible strategy $\pi^{\star}(\cdot) \in \mathbb{A}$ such that 
\begin{itemize}
    \item the SDE system \eqref{eq: x}-\eqref{eq: z_c} admits unique in distribution solution $(\mathcal{X}, \mathcal{Z})$,
    \item the fixed point condition is $\mu = \text{Law}(V | B)$, 
    \item given $J(\cdot)$ defined in \eqref{jinfeq}, it satisfies
\begin{equation}
    \label{mfu}
J(\pi^{\star}) = \inf_{ \pi \in \mathbb{A}} J(\pi).
\end{equation}
\end{itemize}
\end{definition}
{Note that we consider the conditional path law $\mu(\omega) \in \mathcal{P}_2(C([0,T]; \mathbb R_+))$ for almost every $\omega$, as we will see later in the paper, we are particularly interested in the Wasserstein-2 distance between the probability measures of the
wealth processes under different market setups on $(C([0,T]; \mathbb R_+), \|\cdot\|)$, where $\|x\| = \sup_{t \in [0,T]} |x_t|$.} 
{In fact, with sufficient regularities of market coefficients (Assumption~\ref{asmp: lipfn}), we can use \cite[Theorem 1.2]{carmona2016lectures} to conclude that $\E[\sup_{0 \leq t \leq T} |X_i(t)|^2] < \infty$, $\E[\sup_{0 \leq t \leq T} |V(t)|^2] < \infty$ for $i = 1, \ldots, n$. Thus, the random variable $\E[|V_t|^2 | \mathcal F^B_t]$ has finite expectation for $t \in [0,T]$, and $\mu(\omega) \in \mathcal{P}_2(C([0,T]; \mathbb R_+))$ for almost every $\omega$.}

Under common noise $B$, $\mu_\cdot$ is a random measure defined as the conditional law of a state $V(t)$ given the realization of common noise. In general, the equilibrium measure $\mu$ may not be adapted to the filtration generated by the initial condition and common noise, and the notion of a weak mean field equilibrium has been introduced in \cite{carmona2016mean, 10.1214/22-AAP1876}. 
However, characterizing the optimal arbitrage quantity as in \ref{f1} requires the assumption $\mathbb{F} = \mathbb{F}^{B}$, the filtration generated by the $n$-dimensional Brownian motion $B(\cdot)$, in order to use martingale representation theorem. For this reason, we define the mean field equilibrium in the strong sense.

We now define the notion of uniqueness of the Nash equilibrium, which we adopt in our paper. Unlike the typical way of characterizing the unique optimal control for Nash equilibrium, in the relative arbitrage problem, investors focus on their optimal wealth and do not require the optimal control to be unique. 
\begin{definition}
\label{mfemu}
We say that the uniqueness holds for the MFG equilibrium if any two conditional path measures, $\mu^a, \mu^b \in \mathcal{P}_2(C([0,T]; \mathbb R_+))$, of wealth processes in Definition~\ref{def: mfe} , defined on the same filtered probability space $(\Omega, \mathcal{F}, \mathbb{F}, \mathbb{P})$, with the same initial law $\text{\boldmath$\mu$}_0 \in \mathcal{P}_2(\mathbb{R}_+)$, 
\begin{equation}
\label{muequal}
\mathbb{P}[\mu^a = \mu^b] = 1.    
\end{equation}
\end{definition}
In this case, $\{\text{\boldmath$\mu$}^a_t\}_{t \in [0,T]}$ and $\{\text{\boldmath$\mu$}^b_t\}_{t \in [0,T]}$ are indistinguishable processes in $C([0,T]; \mathcal{P}_2(\mathbb{R_+}))$. That is, $\text{\boldmath$\mu$}^a_t(\omega) = \text{\boldmath$\mu$}^b_t(\omega)$ for all $t$ and all $\omega$ outside a null set.

\subsection{Searching optimal strategies through fixed-point problems}

{To this end, we summarize the coupling relationship between market dynamics and the decision making of competitive investors. {As can be seen from previous discussion, the evolution of market dynamics are controlled by (the distribution of) the investment strategies, and the optimal strategy needs to be solved at the Mean Field Equilibrium. Thus, we search the optimal strategy through fixed-point arguments -- a strategy function is fixed and we solve for the optimal arbitrage quantity under the specified market dynamics. The optimal investment strategies are found if the corresponding strategy of optimal arbitrage coincide with the strategy we fixed at the very beginning. This is specified in the following steps.}

\bigskip 

For the sake of simplicity in notation, we denote $(\mathbf{x}, \mathbf{z}) := (\mathcal{X}(t), \mathcal{Z}(t))$ in the following statements and derivations. 

\textbf{Solving Mean Field Game solution over $[t,T]$, for every $t \in [0,T)$}
\begin{enumerate}[label=(\roman*)]
\item  We start with a given control process for the representative player $\pi \in \mathcal{F}^{MF}$. This determines the coefficients $\gamma(\cdot)$ and $\tau(\cdot)$, and thus the corresponding process $\mathbf{z}$ following \eqref{eq: z_c}. 
\item Solve the McKean-Vlasov system \eqref{eq: x}-\eqref{eq: z_c} to obtain the unique solution $(X, \mu, B, v_0)$. 
Define 
\begin{equation}
\label{eq: littleu}
    U(T-t) := \E \left[ \mathcal{V}^{\star}(T) L^{\star}(T) | \mathcal{F}^{MF}_t \right]\, /\, \mathcal{V}^{\star}(t).
\end{equation} 
Thus, the optimal initial proportion of the player $\ell$ and a representative player to achieve relative arbitrage can be characterized as $u^{\ell}(T-t,\mathbf{x}, \mathbf{z}) = e^{c_{\ell}} U(T-t)$ and $u(T-t,\mathbf{x}, \mathbf{z}) = e^{c} U(T-t)$, respectively.

{
\item With the optimum $\{u(T-t)\}_{t \in [0,T], \ell = 1, \ldots, N}$, where $u(T-t) := \inf_{\pi \in \mathbb{A}} J(\pi)$, we determine the corresponding optimal control $\pi^{\star} := \arg \inf_{\pi \in \mathbb{A}} J(\pi)$. Intuitively, the players seek the best response with respect to the other players. We define the solution map $\Phi$ such that $\pi^{\star} = \Phi(\pi)$. 
\item Find $\widehat{\pi}$, such that $\widehat{\pi} = \Phi(\widehat{\pi})$. The optimization problem in the previous step is thus solved through a fixed point problem in the path space of strategies. We specify the existence of such $\Phi(\cdot)$ in Theorem~\ref{mfeuniq}.
}
\end{enumerate}

{Note that the steps to seek equilibrium we introduce is different from the extended mean field game with joint measure of state and control is formulated in \cite{carmona2018probabilistic}. The existence and uniqueness result is provided in \cite[Theorem 4.65]{carmona2018probabilistic} for models with constant volatility and with the drift function affine in state and control. Here, a modified version of extended mean field game is discussed, where the state processes and cost functional depend on different measures, and uniqueness of Nash equilibrium is specified. Intuitively, }
an investor chooses a strategy profile, then solves for the current market \eqref{eq: x}-\eqref{eq: z_c}. {That is, the coefficients $(\gamma, \tau)$ are determined with the chosen strategy profile.}  With the diffusion coefficients in the market dynamics, we find the optimal arbitrage quantity \eqref{ggg}, which can be solved as the minimal nonnegative solution of \eqref{inequ}-\eqref{eqsol}. When $0<\delta \leq 1$, we have
\begin{equation}
\label{eq: fx-u0}
u(T-t, \mathbf{x}, \mathbf{z}) \mathcal{V}(t) = V(t), \quad (\mathbf{x}, \mathbf{z}) :=(\mathcal{X}(t), \mathcal{Z}(t)).
\end{equation}
Let $\mathbf{1}$ be the $n \times 1$ column vector that has all $n$ elements equal to one. By \eqref{eq: fx-u0}, it follows that $\mathbf{z} \cdot\mathbf{1} = \E[u^{\ell}(T-t, \mathbf{x}, \mathbf{z}) (\delta \mathbf{x} \cdot \mathbf{1} + (1-\delta) \mathbf{z} \cdot\mathbf{1})|\mathcal{F}_t^B]$ is $\mathcal{F}_t^B$-measurable. With $1 - (1-\delta) \E[u^{\ell}(T-t, \mathbf{x}, \mathbf{z})  | \mathcal{F}_t^B] > 0$, we get $\mathbf{z} \cdot\mathbf{1} = \frac{\delta \E[u^{\ell}(T-t, \mathbf{x}, \mathbf{z})  X(t) | \mathcal{F}_t^B]}{1 - (1-\delta) \E[u^{\ell}(T-t, \mathbf{x}, \mathbf{z})  | \mathcal{F}_t^B]}$.
This leads to {
\begin{equation}
\label{v0v0}
\begin{aligned}
        \mathcal{V}(t) &= \frac{\delta X(t)}{1-{(1-\delta)} u(T-t, \mathbf{x}, \mathbf{z})}, \\
    \text{and} \quad V(t) &= \frac{ u(T-t, \mathbf{x}, \mathbf{z}) \delta X(t)}{1- {(1-\delta)} \E[ u(T-t, \mathbf{x}, \mathbf{z})|\mathcal{F}_t^B ]},
\end{aligned}
\end{equation}
}
where $X(t)$ is the total capitalization defined in \eqref{Mportfolio-2}. We thus determine $\mathbf{z}$ from $V(\cdot)$ and $\pi(\cdot)$. The next step is to use this information obtained above to optimize relative arbitrage opportunities and solve for the corresponding optimal strategies. If the result is consistent with the initial guess of one's strategy profile, we obtain a Nash equilibrium. In the next section, we elaborate on the details of finding the equilibrium. We summarize the above flow in Figure~\ref{fig: fixedpt} below.}
{
\begin{figure}[ht] 
  \centering 
\begin{tikzpicture}[node distance = 3.5 cm, auto]

  \node (A) {$\pi(\cdot)$};
  \node (B) [right of=A] {$(\mathcal X, \mathcal{Z})$}; 
    \node (C)[right of=B] {$u(T-\cdot, \mathbf{x}, \mathbf{z})$};
  \draw[->] (A) -- (B) node[midway] {\eqref{eq: x}-\eqref{eq: z_c}}; 
  \draw[->] (B) -- (C) node[midway] {\eqref{inequ}-\eqref{eqsol}};
  \draw[->,bend left=12] (C) to node[midway, fill=white] {solve} (A);
\end{tikzpicture}
  \caption{The formulation of the fixed point problems.  {Note that to go from $(\mathcal X, \mathcal{Z})$ to $u(T-\cdot, \mathbf{x}, \mathbf{z})$, only the volatility of the market dynamics $(s(\cdot), \tau(\cdot))$ is needed to solve \eqref{inequ}-\eqref{eqsol}. For the arrow from $u(T-\cdot, \mathbf{x}, \mathbf{z})$ to get back to $\pi(\cdot)$, we elaborate in Section \ref{sec: solveNE} how to develop a consistency condition to check if the optimal strategy coincide with $\pi(\cdot)$ we started. The fixed point problem can start from any node (i.e., $\pi(\cdot)$, $u(T-\cdot, \mathbf{x}, \mathbf{z})$, or $(\mathcal X, \mathcal{Z})$) of this diagram by freezing the quantity of the node, compute the rest of the quantities sequentially according to the arrow of the diagram until reaching back to the node we started from. The fixed point solution corresponds to the quantity of the node that satisfies the consistency condition -- the frozen quantity as a starting point of the flow coincides with the quantity of the same node that is computed through the whole diagram.} } 
  \label{fig: fixedpt} 
\end{figure}
}

\subsection{Equilibrium based on functional generated portfolios}
\label{sec: solveNE}

Note that the fixed point mapping $\Phi(\cdot)$ is generally not unique and there can be multiple solutions for the optimal strategies that satisfy the unique mean field equilibrium in Definition~\ref{mfemu}. We use the following theorem to find the functional generated portfolio type of strategies and focus on the equilibrium when investors adopt these type of strategies. The portfolios generated by functions are proposed in \cite{fernholz1999portfolio}, and the functionally generated portfolio in a multi-player regime is discussed in \cite{yang2021topics}. {The optimal strategy is implicitly defined as the solution of the following fixed point problem, where the coefficients $\gamma(\cdot)$ and $\tau(\cdot)$ in $\mathcal{Z}(\cdot)$ depend on the strategy.} The proof details can be found in Appendix~\ref{appendix:proofs}.


\begin{thm}[Fixed point problem]
\label{mfeuniq}
Fix $\delta \in (0,1]$. Assume $u(\cdot) \in C^{1,3,3}([0,T]\times \mathbb R^n_+ \times \mathbb R^n_+)$. 
Under Assumption~\ref{xv}-\ref{hh}, the best response map $\Phi(\cdot)$ gives the following 
fixed point problem $\pi^{\star} = \Phi(\pi^{\star})$. 
{At each time $t$, $(\mathbf x, \mathbf z) = (\mathcal X(t), \mathcal Z(t))$,
\begin{equation}
\label{eq:bestresponse}
\begin{aligned}
\Phi(\pi)(t) & = \mathcal{X}(t) D_x \log u^{\pi}(T-t, \mathbf{x}, \mathbf{z})+ \mathcal Z(t) D_{z}  \log u^{\pi}(T-t, \mathbf{x}, \mathbf{z}) + \Pi(t).
\end{aligned}    
\end{equation}
In particular, the optimal strategy solved from the fixed point problem is a combination of capitalizations, trading volumes, and the benchmark portfolio \eqref{bigpi}. That is, for each $i = 1, \ldots, n$,
\begin{equation}
\label{anotherpi2}
\begin{aligned}
\pi^{\star}_i(t) &= X_i(t) D_{x_i} \log u(T-t, \mathbf{x}, \mathbf{z}) \\
& \quad {} + \mathcal Z_i(t)D_{z_i}  \log u(T-t, \mathbf{x}, \mathbf{z})\vert_{(\mathbf x, \mathbf z) = (\mathcal X(t), \mathcal Z(t))} + \Pi^{\star}_i(t)\bigg \vert_{(\mathbf x, \mathbf z) = (\mathcal X(t), \mathcal Z(t))}.
\end{aligned}
\end{equation}
}
At terminal time $t = T$, $\pi_i^{\star}(T) = \Pi^{\star}_i(T)$. $u(T, \mathbf{x}, \mathbf{z})$ reads as the value function of the mean field game under the equilibrium. 

\end{thm}

We can show in the same vein as \eqref{v0v0} that {for $u(T-t,\mathbf{x}, \mathbf{z}) =: e^{c} U(T-t,\mathbf{x}, \mathbf{z})$,
\begin{equation}
\label{eq: ufixedt}
   \begin{aligned}
           U(T-t) & = \mathbb{E}^{\mathbf{x}, \mathbf{z}} \big[ \mathcal{V}(T-t) L(T-t) | \mathcal{F}^{MF}_t \big]\, /\, \mathcal{V}(0)\\
           & = \frac{\delta}{\mathcal{V}(0)} \E^{\mathbf{x}, \mathbf{z}} \left[ \frac{\widehat{X}(T-t)}{1- {(1-\delta) e^c} \E[U(t) | \mathcal{F}_t^B]} \Bigg \vert \mathcal{F}^{MF}_t \right] .
   \end{aligned}
\end{equation}}
At Nash equilibrium, this leads to the fixed-point condition on the path space of optimal arbitrage quantity.
If the fixed point solution $U$
is unique, then Nash equilibrium $\mu^{\star} \in \mathcal{P}_2(C([0,T], \mathbb{R}_+))$ is unique in the sense of Definition~\ref{mfemu}. In the following, we provide the conditions under which the Nash equilibrium is unique.

We first state an assumption that leads to the uniqueness result below.
\begin{asmp}
\label{asmp:decouplex}
    For $t \in [0,T]$, $\tilde{\tau} := T-t$, denote the deflated market capitalization \eqref{xhat} as $\widehat{X}^U(\tilde{\tau}) := X^U(\tilde{\tau}) L^U(\tilde{\tau})$ where we emphasize in the superscript the dependence of the deflated capitalization on the optimal arbitrage quantity $U(\cdot)$. We assume that for every $u, v \in \mathcal U$, it satisfies 
    \begin{equation}
    \label{eq: asmp3}
    \sup_{\tilde{\tau} \in [0,T], \mathbf{x},\mathbf{z} \in \mathbb R^n_+ \times \mathbb R^n_+ }\E^{\mathbf{x}, \mathbf{z}} \left [ \lvert \widehat{X}^u(\tilde{\tau}) - \widehat{X}^v(\tilde{\tau}) \rvert^2 \bigg| \mathcal{F}_{t}^{MF}  \right] < M \lVert u-v\rVert_{\mathcal U}^2 
    \end{equation}
    for some constant $M > 0$. 
\end{asmp}

 Intuitively, Assumption~\ref{asmp:decouplex} means that the deflated market capitalization responds to the changes in the optimal arbitrage quantity (and thus the changes in the initial wealth of investors) in a stable manner.

\begin{thm}[Uniqueness of Nash equilibrium]
\label{thm:unique}
Under Assumption~\ref{xv}-\ref{asmp:decouplex}, consider the sub-problems $u(T-t)$ at every $t \in [0,T]$. Starting at time $t \in [0,T)$, take $u, v \in (0,1)$ as the different values of the initial relative arbitrage quantity of the investor $\ell$, as defined in \eqref{ggg}. Nash equilibrium $(\pi^{\star}, \mu^{\star})$ is unique when 

\begin{equation}
\label{eq: condforn}
\frac{1-\delta^2}{\delta} \overline{e^c}  \in (0,1),
 \quad M < x_0 \frac{\delta + \overline{e^{c}}(\delta^2 -1)}{1-(1-\delta) \overline{e^{c}}}.
\end{equation}
\end{thm}

The proof (Appendix~\ref{appendix:proofs}) starts with the justification for the above connection between fixed point problems in the strategy space and proves a contraction of the fixed point mapping in the path space of optimal arbitrage quantities.

\subsection{Examples}

To this end, we provide the following example to show a viable optimal arbitrage solution in a realistic stock market system.

\begin{example}[Admissible strategy]
    If the market coefficients defined in Assumption~\ref{xv} are of the form
    \[
    (b_i, s_{ik}, \gamma_i, \tau_{ik})(\mathcal{X}(t), \mathcal{Z}(t)) = (\widetilde{b}_i, \widetilde{s}_{ik}, \widetilde{\gamma}_i, \widetilde{\tau}_{ik})(\mathbf{m}(t), \mathbf{n}(t)),
    \]
    where $\mathbf{m}(\cdot)$ is the market portfolio and $\mathbf{n}_i(\cdot) := \frac{\mathcal{Z}_i(\cdot)}{Z(\cdot)}$, $\mathbf{n}(\cdot) := (\mathbf{n}_1(\cdot), \ldots, \mathbf{n}_n(\cdot))$.
Then 
\[
\begin{aligned}
    X_i(t) D_{x_i} \log u^{\star}(T-t, \mathbf{x}, \mathbf{z}) = \mathbf{m}_i \Big(D_{m_i} &\log \widetilde{U}(T-t, \mathbf{m}(t), \mathbf{n}(t)) \\
    &- \sum_{j=1}^n \mathbf{m}_j D_{m_j} \log \widetilde{U}(T-t, \mathbf{m}(t), \mathbf{n}(t))\Big),
\end{aligned}
\]
\[
\begin{aligned}
\mathcal Z_i(t) D_{z_i} \log u^{\star}(T-t, \mathbf{x}, \mathbf{z}) = \mathbf{n}_i \Big(D_{n_i} &\log \widetilde{U}(T-t,\mathbf{m}(t), \mathbf{n}(t)) \\
&- \sum_{j=1}^n \mathbf{n}_j D_{n_j} \log \widetilde{U}(T-t,\mathbf{m}(t), \mathbf{n}(t))\Big),
\end{aligned}
\]
in \eqref{anotherpi2}. Hence, assume $\widetilde{U}(\cdot) \in C^{1,3,3}([0,T] \times \mathbb{R}^ n_+ \times \mathbb{R}^n_+)$, we get from \eqref{anotherpi2} that the optimal strategy in equilibrium satisfies $\sum_{i=1}^n \pi^{\star}_i(t) = 1$. 
\end{example}

\begin{example}
\label{ex: vsm}
We construct the capitalization coefficients of stocks using this similar idea in volatility-stabilized market models (\cite{fernholz2005relative}). For $1 \le i, j \le n$, with infinite number of investors,
\begin{equation}
\label{eq: coeff-vsm}
    \beta_i(t) = \frac{C_x}{\mathbf{m}_i(t) \mathcal Z_i(t)} , \quad a_{ij} = \frac{X_i(t)}{\mathcal Z_i(t)} X(t)\delta_{ij},
\end{equation}
where $\delta_{ij} = 1$, when $i = j$; and $\delta_{ij} = 0$ otherwise, when $i \neq j$. $\mathcal V_0 = \frac{x_0}{2 - \frac{1}{N} \sum_{k=1}^N u^k(T)}$, and $y_0 = \frac{x_0}{2 - \frac{1}{N} \sum_{k=1}^N u^k(T)}- \frac{x_0}{2}$. This construction possess the \textit{leverage effect} in volatility-stabilized market models, where the rate of return and volatility have a negative correlation with the capitalization of the stock relative to the market $\{\mathbf{m}_i(t)\}_{i = 1, \ldots, n}$. 
Smaller stocks tend to have higher volatility than larger stocks. Further, Assumption~\ref{hasmp} is satisfied as \[
D_i H(\mathbf{x}) := \frac{b_i(\mathbf{x},\mathbf{z})}{a_{ii}(\mathbf{x},\mathbf{z})}= \frac{X_i(t) \beta_i(\mathbf{x},\mathbf{z})}{a_{ii}(\mathbf{x},\mathbf{z})} = \frac{C_x}{x_i},
\]
\[k(\mathbf{x},\mathbf{z}) := - \sum_{i=1}^n \sum_{j=1}^n \frac{a_{ij}( \mathbf{x}, \mathbf{z})}{2} [ D_{ij}^2 H( \mathbf{x}) + D_i H ( \mathbf{x},\mathbf{z}) D_j H( \mathbf{x})] = 0,\]
and the market price of risk follows
\[\theta_i(\mathcal{X}(t),\mathcal{Z}(t)) = s_{ii}'(\mathcal{X}(t),\mathcal{Z}(t)) D_i H(\mathcal{X}(t)) = C_x ( X(t))^{\frac{1}{2}} (X_i(t) \mathcal Z_i(t))^{-\frac{1}{2}},\] $i = 1, \ldots, n,$
with $H( \mathbf{x}) = C_x \sum_{i=1}^n \log x_i$. $L(t) = \prod_{j=1}^n \frac{x_j}{X_j(t)}$, $\widehat{X}(t) = X(t)\prod_{j=1}^n \frac{x_j}{X_j(t)}.$
By \eqref{ggg},
\begin{equation}
\label{eq: uexample}
\begin{aligned}
u^{\ell}(T-t) = & \frac{e^{c_{\ell}} X_1(t) \ldots X_n(t)}{\mathcal{V}(t)}\E \left[\frac{\mathcal{V}(T)}{X_1(T) \ldots X_n(T)} \Big| \mathcal{F}^{MF}_t \right]\\
 = & \frac{\delta e^{c_{\ell}} X_1(t) \ldots X_n(t)}{\mathcal{V}(t) (1-(1-\delta)\overline{e^c})}\E \left[\frac{X(T)}{X_1(T) \ldots X_n(T)} \Big| \mathcal{F}^{MF}_t \right].
\end{aligned}
\end{equation}

If all the investors are homogeneous in the game, that is, they have the exact same preference level $c_{\ell} = c$ for all $\ell \in \mathbb{N}$, then any player is a representative player and their optimal strategy is the same. In this case, the strategy can be rewritten as a function of the stock capitalization and the (common) wealth process, $Z(\cdot) = \E[V(\cdot)|\mathcal{F}_t^B] = V(\cdot)$, and the coefficient $\tau_i(\mathbf{x}, v) = v \phi_i(\mathbf{x}, v) \sigma(\mathbf{x}, v)$. Hence, it follows
    \begin{equation*}
        \begin{split}
    \phi_i( \mathbf{x}, v) & = \mathbf{x}_i D_{x_i} \log u^{\ell \star}(T-t, \mathbf{x}, v) \\
    & \quad {} + v \phi_i( \mathbf{x}, v) D_{v}  \log u^{\ell \star}(T-t, \mathbf{x}, v) +\frac{\delta \mathbf x_i + (1-\delta)v}{\delta \mathbf{x} \cdot \mathbf{1}+ (1-\delta)v}.
\end{split}
\end{equation*}
Solve the fixed point problem above, we have 
\begin{equation}
    \label{phi_cor2}
    \begin{aligned}
    \phi^{\star}_i( \mathbf{x}, v) = \frac{\mathbf{x}_i(t) D_{x_i} \log u^{\ell \star}(T-t, \mathbf{x}, v) +\frac{\delta \mathbf{x}_i+ (1-\delta)v}{ \delta \mathbf{x} \cdot \mathbf{1}+ (1-\delta)v}}{1 - v D_{v}  \log u^{\ell \star}(T-t, \mathbf{x}, v)}.
\end{aligned}
\end{equation}

\end{example}

\section{Approximating finite player games of relative arbitrage optimization}
\label{lims}

{In this section, we consider a more realistic problem with finitely many players under Assumption~\ref{xv}. We construct a general infinite interacting system using conditional McKean-Valsov equations as well as a finite counterpart of the market, without a-priori assumption on the exchangeability of the particles. That is, we investigate whether the finite player game with every player interacts with each other can be simplified to the situation of a representative player interacting with the entire distribution of players despite the heterogeneity nature of players.
Then, we study the mean field approximation of the finite system and derive propagation of chaos result to show that the Nash equilibrium of the finite-player game converges to the mean field euilibrium.
}

{\begin{definition}[Mean-field interaction measures]
\label{def: mut}
With a given initial condition $\mu_0 \in \mathcal{P}_2(\mathbb{R}_+)$ from which the initial wealth $v_0$ are sampled,
we define the conditional law of $V(t)$ given $\mathcal{F}^B_t$ as
\begin{equation}
\label{muemf}
    \mu_t := Law \left(V(t) \big| \mathcal{F}_t^B \right), \quad t \ge 0 
\end{equation}
and the conditional law of $(V(t), \pi(t))$, given $\mathcal{F}^B_t$,
is
\begin{equation}
\label{nuemf}
\nu_t := Law(V(t), \pi(t)|\mathcal{F}_t^B) , \quad t \ge  0 . 
\end{equation}
 In particular, the trading volume $\mathcal Z(t)$ can be formulated as $\mathcal Z(t) = \int_{\mathbb{R}_+ \times \mathbb A} x y \, \nu_t (dx \times d y) $, $t \ge 0 $, where $x$ represents the wealth $V(t)$ and $y$ represents the strategies defined in the admissible set $\pi(t) \in \mathbb A$.
\end{definition}
}

{
We consider the market under the Markovian model as in Assumption~\ref{xv}\ref{thomo}. If $\mu_t$, $\nu_t$ are not affected by the investors, then investors face identical and independent optimization problems since their underlying dynamics and their benchmarks are i.i.d. In this case, we can then consider a representative player whose wealth $V(t)$ is generated from a strategy $\pi(t) \in \mathbb{A}$ through \eqref{wealth}. Denote $\mathcal{X} (t) = (X_1(t), \ldots, X_n(t))$, 
\[
    dV^{\ell}(t) = V^{\ell}(t) \left( \sum_{i=1}^{n}\pi^{\ell}_i(t) \beta_{i}( \mathcal X(t), \nu_t)dt + \sum_{i=1}^{n} \sum_{k=1}^{n} \pi^{\ell}_i(t) \sigma_{ik}(\mathcal X(t), \nu_t)dB_{k}(t) \right); 
\]
with $ V^{\ell}(0) = v^{\ell}$. 
We discuss the viability of simplifying the problem as a representative player in Section~\ref{limitsection}. Let us consider 
\[
(\mathcal{X}, V, \nu, B) \in (C([0,T], \mathbb{R}^n_+), C([0,T], \mathbb{R}_+), \mathcal{P}_2(C([0,T], \mathbb{R}_+ \times \mathbb{A})), C([0,T], \mathbb{R}^n)).
\]
The goal for this section is to study the optimization of relative arbitrage opportunities in the conditional McKean-Vlasov system for a representative player, that is, in the system
\begin{equation}
\label{mcvsde}
    d \mathcal{X}(t) = \mathcal{X}(t)\beta(\mathcal{X}(t), \nu_t)dt + \mathcal{X}(t) \sigma(\mathcal{X}(t), \nu_t)dB(t), \quad \mathcal{X}_0 = \mathbf{x};
\end{equation}
\begin{equation}
\label{mcvsde2}
    dV(t) = V(t) \pi'(t) \beta(\mathcal{X}(t), \nu_t)dt + V(t) \pi'(t) \sigma(\mathcal{X}(t), \nu_t)dB(t), \quad V(0) = v_0,
\end{equation}
where $B(\cdot) = (B_1(\cdot), \ldots, B_n(\cdot))$ is $n$-dimensional Brownian motion. $\nu$ is as defined in Definition~\ref{def: mut}. 
We can show the well-posedness of the McKean-Vlasov system under standard assumptions on market coefficients. The formal statements and the proof can be found in Appendix~\ref{appendix: nusystem}.
}

\subsection{{The limit of finite dynamical systems}}
\label{limitsection}

{In reality, there is a finite number of investors in the market and one would wonder if the aforementioned mean-field formulation is a valid model for a finite system.}

{To this end, we define the finite counterpart of the model setup. We do not require the $N$-player system to share the same Wiener process across $N$. Indeed, the systemic effect of random noises towards the market might be different when we consider a finite or infinite group of investors interacting with the market. We let Wiener process} $W^N \in C([0, T], \mathbb R^n)$ be the common noises for the $N$-player system \eqref{nparticle}-\eqref{eq: nv}, and the filtered probability space $(\Omega, \mathcal{F}, \mathbb{P})$ support $W^N$ as well. {This is more realistic as we do not need to fix the exact same trajectories of noise for every system, and the finite system is not merely a truncation of the given infinite system under the same driving Wiener process across $N$.}
We assume that $B$ and $W^N$ for each $N$ are adapted to $(\mathcal F_t^{MF})_{t \in [0,T]}$, are independent of each other, and are independent of $\mathbf{x}_0$ and $\mathbf{x}_0^N$. The filtration $\left( \mathcal{F}_t \right)_{t \in [0,T]}$ is the natural filtration generated by $W^N$ and $B$.

We first define the empirical measure in the finite-particle system 
.
\begin{definition} [Empirical measures in the finite $N$-particle system]
\label{nunt}
Consider $\mathcal{F}$-measurable $ C([0,T]; \mathbb{R}_+) \times C([0,T]; \mathbb{A})$-valued  random variables $(V^{N, \ell}, \pi^{\ell})$ for every investor $\ell = 1, \ldots, N$. We define empirical measures $\nu^N \in \mathcal{P}_2(C([0,T], \mathbb{R}_+) \times C([0,T], \mathbb{A})) \cong \mathcal{P}_2(C([0,T], \mathbb{R}_+ \times \mathbb{A}))$ of the random vectors $(V^{N, \ell}(t), \pi^{\ell}(t))$ as
\begin{equation}
\nu^N_t := \frac{1}{N} \sum_{\ell=1}^N {\bm \delta}_{(V^{N, \ell}(t), \pi^{\ell}(t))}, \quad  t \ge 0, 
\end{equation}
where ${\bm \delta}_x$ is the Dirac delta mass at $x \in \mathbb{R}_+ \times \mathbb{A}$. Thus for any Borel set $A \subset \mathbb R_{+} \times \mathbb A$,
\begin{equation}
\nu^N_t(A) = \frac{1}{N} \sum_{\ell=1}^N {\bm \delta}_{(V^{N, \ell}(t), \pi^{\ell}(t))} (A)  = \frac{1}{N} \cdot \# \{\ell \leq N: (V^{N, \ell}(t), \pi^{\ell}(t)) \in A\} , 
\end{equation}
where $\# \{\cdot\}$ represents the cardinality of the set.
{In the same vein, for $t\ge 0 $, we define the empirical measures $\mu^N \in \mathcal{P}_2(C([0,T], \mathbb{R}_+))$ of the wealth by 
\begin{equation} \label{eq: NE def2}
\mu_t^N := \frac{1}{N} \sum_{\ell=1}^N {\bm \delta}_{V^{\ell}(t)}.
\end{equation}
}
\end{definition}

For a fixed $N$, with $\nu^N_t$ in Definition~\ref{nunt} that generalizes the interaction term $\mathcal{Y}(t)$, we can construct a general market system as
\begin{equation}
\label{nparticle}
 d\mathcal{X}^N(t) = \mathcal{X}^N(t) \beta( \mathcal{X}^N(t), \nu^N_t)dt + \mathcal{X}^N(t) \sigma( \mathcal{X}^N(t), \nu^N_t)dW^N(t); 
\end{equation}
with $\mathcal{X}^N(0) = \mathbf{x}^N_0$ 
and for $\ell = 1, \ldots , N$, $V^{N, \ell}(0) = v^{N, \ell}$,
\begin{equation} \label{eq: nv}
\begin{split}
   \hspace{-3cm}  dV^{N, \ell}(t) & = V^{N, \ell}(t) \Big( \sum_{i=1}^{n}\pi^{\ell}_i(t) \beta_{i}(\mathcal X(t), \nu^N_t)dt  \\
    & \hspace{3cm} {} + \sum_{i=1}^{n} \sum_{k=1}^{n} \pi^{\ell}_i(t) \sigma_{ik}(\mathcal X(t), \nu^N_t)dW^N_{k}(t) \Big).
    \end{split}
\end{equation}

We first show the convergence of the continuous stochastic process $\mathcal{X}^N(t)$.

\begin{prop}
\label{xlimit}
Given an $n$-dimensional vector-valued process $\mathcal{X}^N(t)$ \eqref{nparticle} for $N \in \mathbb{N}$, whose coefficients satisfy Assumption~\ref{xv}. There exists a $n$ dimensional continuous process $\mathcal{X}$ defined on the probability space $(\Omega, \mathcal{F}, \mathbb{P})$, such that $\{\mathbb{P}^{\mathcal{X}^N}\}$ converges weakly to $\{\mathbb{P}^{\mathcal{X}}\}$ as $N \rightarrow \infty$. 
\end{prop}

{
We can also show this result for $\{V^{N, \ell}, \mu^N, \nu^N, W^N\}_{N \in \mathbb N}$. With $\pi \in \mathbb A$, by Prokhorov's theorem and \cite{sznitman1991topics}, we can show the tightness of $\{\mu_N\}_{N \geq 1}$ and $\{\nu_N\}_{N \geq 1}$ on $\mathcal{P}_2(C([0,T]; \mathbb{R}_+))$ and $\mathcal{P}_2(C([0,T]; \mathbb{R}_+))$, respectively, equipped with the Wasserstein distance.$\{\mathbb{P} \circ (W^N)^{-1}\}_{N \geq 1}$'s tightness can be shown by Kolmogorov's criterion. We provide the streamlined proof in Appendix~\ref{appendix: proof5}.}
\begin{prop}
\label{prop: tightv}
Under Assumption~\ref{xv}-\ref{asmp: lipfn}, the family $\{V^{N, \ell}, \mu^N, \nu^N, W^N\}_{N \in \mathbb N}$ is tight for each $\ell = 1,\ldots, N$, $N \in \mathbb N$.
\end{prop}


\subsection{Approximations of \texorpdfstring{$N$}{Lg}-player Nash equilibrium}
\label{nemfe}

{
We have mentioned that in our extended mean field game formulation, we pursue the equilibrium for a so-called representative player, which is a generic player in the infinite-population. It is a straightforward result from the law of large numbers that the mean field measures are the weak limit of the empirical measure in finite player game when the players are exchangeable. That is, the interacting particle system stays the same after permuting the index for the players. However, because of the heterogeneity of the players in this paper, the exchangeability of each player's state is not straightforward. Thus, in this section we first provide the propagation of chaos result to justify the use of the representative player and the formulation of the mean field game problems as the limits of their finite counterparts. 
In the end, we prove that the mean field equilibrium is an appropriate approximation of $N$-player relative arbitrage problem by showing that the MFE we obtain agrees with the limit of the finite equilibrium.
}

{We denote $\mathcal{C}^{A} = C([0,T]; \mathbb{R}_+ \times \mathbb{A})$ as the path space equipped with the supremum norm $\|x\| = \sup_{t \in [0,T]} |x_t|$. }
{As solving a mean field game is computationally much easier than solving a finite player game, the idea is to check whether the optimal strategy function solved at mean field equilibrium can be used for $N$-player game. We investigate the relationship between the mean field measures at equilibrium and the empirical measures in the below theorem.}

{Under Assumptions~\ref{xv}- \ref{asmp: lipfn}, let $\phi : \mathbb{R}_+ \times \mathcal{P}_2(\mathbb{R}_+ \times \mathbb{A}) \times \mathbb{R} \rightarrow \mathbb{A}$ as the optimal strategy function following \eqref{anotherpi2}. We denote $\mu^{\star}$ and $\nu^{\star}$ as the conditional path laws in Definition~\ref{def: mut} at mean field equilibrium, where the strategies follow $\pi^{\star}_i(t) = \phi(\mathcal{X}(t), \nu_t, c)$. We denote $\mu^{N \star}$ and $\nu^{N \star}$ as the empirical measures (Definition~\ref{nunt}) of the $N$-player game counterpart with the same strategy function $\pi^{\ell \star,N} = \phi(\mathcal{X}^N, \nu^N, c^{\ell})$. }  

\begin{thm}[Propagation of chaos result for the conditional laws]
\label{tight}
{
In addition to Assumptions~\ref{xv}- \ref{asmp: lipfn}, 
assume the empirical measures share the same initial condition with the mean field counterpart.
Then, there exist $\nu^{\star}_t$ and $\mu^{\star}_t$ exist for $t \in [0,T]$ such that $\E[\mathcal{W}_2^2(\mu^{N{\star}}, \mu^{\star})] \rightarrow 0$ and $\E[\mathcal{W}_2^2(\nu^{N{\star}}, \nu^{\star})] \rightarrow 0$, as $N \to \infty$, where $\mathcal{W}_2$ distances are induced by the supnorm on $C([0,T]; \mathbb{R}_+)$ and $\mathcal{C}^{A}$, respectively.}
\end{thm}

{
See Appendix~\ref{appendix: proof5} for the proof. We can easily weaken the assumption of the same initial conditions as the initial empirical measure to its corresponding mean-field measure in the $W_2$ sense. 
The above results justify that the influence of each player on the whole system is diminishing as $N$ grows to infinity. Asymptotically, the pairwise interactions of every two players can be simplified as the representative player interacting with the market system through the distribution of the states and actions of the population. We can thus consider a representative player and solve a single optimization problem instead.

We now investigate whether the optimal strategy derived in Mean Field games can achieve Nash equilibrium in the $N$-player case. {Consider the solution $u^{\star \ell, N}$ of the Cauchy PDE in the finite player case, that is, \eqref{inequ}-\eqref{eqsol} but with coefficients of the finite-player system. In particular, the coefficients of the trading volume would be different in the finite and mean-field system as the trading volumes are $\frac{1}{N} \sum_{\ell=1}^N V^\ell(\cdot) \pi^{\ell}(\cdot)$ and $\E[V^\ell(\cdot) \pi^{\ell}(\cdot)|\mathcal{F}_t^B]$ and their coefficients are solved using \eqref{eq: nv} and \eqref{mcvsde2}, respectively.} We show below that the optimal arbitrage $u^{\star, \ell, N}$ 
strongly converges to optimal arbitrage value $u^\star$ in the mean field game setting \eqref{mfu}, as $N\to \infty$. Readers can find the details of the proof in Appendix~\ref{appendix: proof5}.

\begin{thm}
\label{netomfe}
Under Assumptions~\ref{xv}-\ref{asmp: lipfn}, if $c_\ell \equiv c$ for some $\ell= 1, \ldots, N$, then $u^{\star}(T-t, \mathbf{x}, \mathbf{z}) = \lim_{N \rightarrow \infty} u^{\star \ell, N}(T-t, \mathbf{x}, \mathbf{z})$ a.s., for initial condition $(\mathbf{x}, \mathbf{z}) \in \mathbb{R}_+^n \times \mathbb{R}^n_+$, $t \in [0,T]$.

\end{thm}

This result shows that the mean field game formulation here is indeed a valid approximation for the optimal arbitrage problem among a large (finite) number of players. 
In particular, the optimal arbitrage function in $N$-player game can be approximated by the optimal arbitrage function used in the mean field game.

\section{Discussions}
\label{sec: discussion}

{
In this work, we consider the competition among investors that pursue a relative arbitrage goal and derive the market as an infinite interacting particle system with common noise. We clarify the randomness introduced from the market, the preference parameter, and the initial market condition, which is not obvious when only looking at the game in a finite-player scenario. We show that the existence and uniqueness of mean field equilibrium following the functional generated portfolio based on stock capitalization and trading volume. The mean field equilibrium and its finite-player counterpart correspond to fixed point problems on the portfolio (path) space, which is solved by an associated Cauchy problem. The benefit of considering the infinite limit of the number of players lies in the simplification of the coupled coefficients in the market dynamics and Cauchy problem by considering a representative player.}

The discussion of a relative arbitrage problem in an interacted system is two-fold: On one hand, in the mean field game regime, every player needs to follow the goal of competing with both the market and other participants, and they all succeeded in adopting the Nash equilibrium strategy. This is our focus of this paper. On the other hand, conditions of the Fichera drift (\cite[Chapter 9]{friedman1975stochastic}) must be satisfied in order for the relative arbitrage opportunity to exist. See \cite{ichiba2020relative} for the Fichera drift condition of the finite player system. If some investor's decision does not satisfy the conditions of Fichera drift, they do not obtain the relative arbitrage opportunity, i.e., $u^{\ell} = e^{c_{\ell}}$). Note that whether the Fichera drift condition is satisfied or not does not affect the achievement of Nash equilibrium. {Designing numerical algorithms, potentially extending the approach in \cite{yang2024finding}, to solve the mean-field game that satisfies the Fichera drift condition remains an avenue for future work.}

A natural extension is to consider a different optimization concept, where a global planner decides on the optimal feedback strategy for all investors. This is a stochastic control problem with the underlying dynamics mean-field interactions. When $N \rightarrow \infty$, this is known as the mean field control problem. This type of problem can be regarded as the problem of a single player who tries to optimize a cost involving the law of her own state, which evolves with a conditional McKean-Vlasov (MKV) dynamics. 
Through the difference of the solution from MFC and MFG, we can be informed about the influence of deviation from other players, and the appropriate information structure to use under different scenario. 

\bigskip

\textbf{Acknowledgements. } We thank anonymous referees and associate editor for their careful reading and helpful suggestions.

\begin{appendices}

\section{Proofs}
\label{sec:appendix}
\subsection{Mean Field Equilibrium}
\label{appendix:proofs}

\begin{lemma}[Gr\"onwall's inequality] \label{gronwall}
Let $h \in L^{\infty}([0,T])$ for some $t_0>0$. Assume that
there exist $a \ge 0$ and $b > 0$ such that
$h(t) \leq a + b \int_0^t h(s) \, {\rm d} s$ for all $t \in [0,T]$. 
Then it holds 
\[
h(t) \leq a \, e^{bt}, \qquad t \in [0,T].
\]
\end{lemma}

\begin{proof}[Proof of Proposition~\ref{pistarr}]
To show $\mathcal{V}(t)$ is a wealth process generated by a strategy, we use \eqref{wealth} and \eqref{bigpi}  to get
\[
\frac{d\mathcal{V}(t)}{\mathcal{V}(t)} = \frac{1}{ \mathcal{V}(t)} \bigg(\delta dX(t) + (1-\delta) \sum_{i=1}^n \mathcal{Z}_i(t) \frac{dX_i(t)}{X_i(t)} \bigg) = \sum_{i=1}^n \Pi_i(t) \frac{dX_i(t)}{X_i(t)} , \quad \text{for}  \ t \in (0,T],
\]
and
\begin{equation} \label{eq:mathcalV(0)}
\mathcal{V}(0) = \delta X(0) + (1-\delta)\E[v], \quad 
\Pi_i(t)
= \frac{\delta X_i(t) + (1-\delta)\mathcal{Z}_i(t)}{\mathcal{V}(t)}.
\end{equation}
Further calculations show that $\Pi_i(t)$ satisfies the self-financing condition \eqref{admpi}. Therefore, we conclude $\Pi(t) \in \mathbb{A}$, since $\sum_{i=1}^N \Pi_i(t) = 1$ and $0 < \Pi_i(t) < 1$, for any $t \in [0,T]$, $i = 1, \ldots, n$.

\end{proof}

\begin{proof}[Proof of Theorem~\ref{mfeuniq}]
When searching for mean field equilibrium, we start from a given choice of $\pi(t) \in \mathbb{A}$, for any $t \in [0,T]$. 
Every player in the mean field game acts optimally by following 
\begin{equation}
\label{mfevstar}
    V^{\star}(t) = \mathcal{V}^{\star}(t) u(T-t),
\end{equation} 
with the rest of the players assumed to be fixed. Thus,
{the Markovian property \eqref{ggg} implies the deflated wealth process
\begin{equation}
    \label{mgVhat}
    \hat{V}^{\ell \star}(t) := V^{\ell \star}(t) L(t)
= \mathbb{E} \big[ \mathcal{V}(T) L(T)|\mathcal{F}^{MF}_t \big]
\end{equation}
is a martingale. As a result, from \eqref{mgVhat}, the $dt$ terms in $d\hat{V}^{\ell}(t) = d(\mathcal V^{\ell}(t) L(t) u^{\ell}(T-t))$ will vanish. In the same vein of the proof of \cite[Theorem 4.1]{ichiba2020relative}, by comparing its diffusion term with the general formula $\hat{V}^{\ell}$ in \eqref{vlhatt}, we derive that at each time $t$, 
\begin{equation}
\label{eq:bestresponse0}
\begin{aligned}
\Phi(\pi^{\ell})(t) & = \mathcal{X}(t) D_x \log u^{\ell, \pi}(T-t, \mathbf{x}, \mathbf{z})+ \mathcal Z(t) D_{z}  \log u^{\ell, \pi}(T-t, \mathbf{x}, \mathbf{z})\\
& \ \ +\frac{\delta X(t)}{\mathcal{V}(t)} \mathbf{m}(t) + \frac{(1-\delta)}{ \mathcal{V}(t)} \E[V(t) \pi(t)|\mathcal{F}^B_t].
\end{aligned}    
\end{equation}
Hence, \eqref{eq:bestresponse} holds.
}

\end{proof}

\begin{proof}[Proof of Theorem~\ref{thm:unique}]

{We consider the optimal arbitrage quantity on the path space. 
When the optimal strategy is the fixed point solution of \eqref{anotherpi2}, we can find a corresponding {fixed point operator $\mathfrak G$ such that} $\mathfrak G(U) = U$, where $U := \left(U(t)\right)_{t \in [0,T]} \in \mathcal{U} 
$. In particular,
\begin{equation}
\label{ufixed1}
\mathfrak G (U)(T-t, \mathbf{x}, \mathbf{z}) = U(T-t,\mathbf{x}, \mathbf{z})
\end{equation}
for every $t \in [0, T]$, $(\mathbf{x}, \mathbf{z}) \in \mathbb{R}_+ \times \mathbb{R}_+$. }

\begin{enumerate}

\item 
 We first set up the fixed-point problem of $u(\cdot)$. We have
\begin{equation}
\label{ute2}
    u(T-t,\mathbf{x}, \mathbf{z}) = e^{c} \mathbb{E}^{\mathbf{x}, \mathbf{z}} \big[ \mathcal{V}(T-t) L(T-t) \big]\, /\, \mathcal{V}(0) = \frac{\delta e^{c}}{\mathcal{V}(0)} \E^{\mathbf{x}, \mathbf{z}} \left[ \frac{\widehat{X}(T-t)}{1- \frac{1-\delta}{N} \E[e^{c}U(t) | \mathcal{F}^B_t]} \right].
\end{equation}
Define $f : \mathcal{U} 
\rightarrow C([0,T]; \mathbb R)$ as follows:
\begin{equation}
\label{fu}
f(u)(t) = \left[1- (1-\delta) \E [ u(t) e^{c} | \mathcal{F}^B]\right]^{-1}, \quad u \in \mathcal{U}, \, t \in [0,T].
\end{equation}
Then, it holds, 
\[
\begin{aligned}
\sup_{\tau \in [0,T], (\mathbf{x}, \mathbf{z}) \in \mathbb{R}^n_+\times \mathbb{R}^n_+} f(u)(\tau,\mathbf{x}, \mathbf{z}) 
& \leq \sup_{\tau \in [0,T]} \left[1- (1-\delta) \E[e^{c}| \mathcal{F}_\tau^B] \right]^{-1}  \\
& \leq \left[1- (1-\delta) E[e^{c}] \right]^{-1} =: C_f.    
\end{aligned}
\]
Define an operator $\mathcal{F}: \mathcal{U} \rightarrow \mathcal{U}$ by
\begin{equation}
 \label{ufixed0}
 \begin{aligned}
\left[ \mathcal{F} * U \right](T-t,\mathbf{x}, \mathbf{z})&= \mathcal{V}^{-1}(0)\E^{\mathbf{x}, \mathbf{z}} \left[ \frac{\widehat{X}(T-t)}{1- (1-\delta)  \E[ e^{c} U(t,\mathbf{x}, \mathbf{z}) | \mathcal{F}^B]} \bigg| \mathcal{F}^{MF}_t \right] \\
&= \mathcal{V}^{-1}(0)\E^{\mathbf{x}, \mathbf{z}} \left[ \widehat{X}(T-t) f (U)(t,\mathbf{x}, \mathbf{z}) \bigg| \mathcal{F}^{MF}_t \right],
 \end{aligned}
\end{equation}
for every $t \in [0,T]$, where $U := \left(U(t)\right)_{t \in [0,T]} \in \mathcal{U}$, $f$ is defined in \eqref{fu}. Consider a mapping $\mathcal{I}: \mathcal{U} \rightarrow \mathcal{U}$, such that $\mathcal I(U)(t, \mathbf{x}, \mathbf{z}) = U(T-t, \mathbf{x}, \mathbf{z})$, for every $t \in [0,T]$. In particular, $\mathcal{I}(\cdot)$ maps the optimal arbitrage quantity for the time horizon $[0,T-t]$ to the optimal arbitrage quantity for the time horizon $[0, t]$. Note that this is different from the subproblems in \eqref{utobj1}, since the time horizon there is $[t,T]$, for every $t \in [0,T]$. We have
\[
\mathcal I(\E^{\mathbf{x}, \mathbf{z}}[\widehat{X}])(T-t,\mathbf{x}, \mathbf{z}) = \E^{\mathbf{x}, \mathbf{z}}[\widehat{X}(t)].
\]
This mapping $\Phi$ is continuous and bounded in $\mathcal{U}$. The boundedness is immediately followed by the bounded nature of $U(\cdot)$. Let $U_m(\cdot)$ be a sequence of functions in $\mathcal{U}$ that converges uniformly to a function $U \in \mathcal{U}$ on $[0,T] \times \mathbb{R}_+ \times \mathbb{R}_+$ as $m \to \infty$. Then for each $t \in [0,T]$,
    \[\lim_{m \to \infty} \mathcal I(U_m)(t) = \lim_{m \to \infty} U_m(T-t) = U(T-t) = \mathcal I(U)(t)\]
    Thus, $\mathcal I$ is continuous and hence, $ \mathcal F*\mathcal I(U)(t) = \mathcal F(U)(T-t)$ for $t \in [0, T]$, $\mathfrak G$ is continuous in the topology of $\mathcal U$. At Nash equilibrium, \eqref{ute2} leads to the fixed-point condition
\begin{equation}
\label{ufixed}
\mathfrak G (U)(T-t, \mathbf{x}, \mathbf{z}) = \left[ \mathcal{F} * \mathcal{I}(U) \right](T-t,\mathbf{x}, \mathbf{z}) = U(T-t,\mathbf{x}, \mathbf{z})
\end{equation}
for every $t \in [0, T]$, $(\mathbf{x}, \mathbf{z}) \in \mathbb{R}_+ \times \mathbb{R}_+$.

\item Denote $\overline{M} :=  M C_f \left(x_0 + \frac{1-\delta}{\delta} \overline{v}_0\right)^{-1} = \frac{M}{\delta x_0}$. 
By triangle's inequality, for $u, v \in \mathcal{U}$,
\[
\begin{aligned}
 \|\mathcal{F} * \mathcal{I}(u)  - \mathcal{F} * \mathcal{I}(v) \|_{\mathcal U} & = \sup_{\tau \in [0,T], (\mathbf{x}, \mathbf{z}) \in \mathbb{R}^n_+\times \mathbb{R}^n_+} \!\!\! \! |\mathcal{F} * \mathcal I(u(\tau,\mathbf{x}, \mathbf{z})) - \mathcal{F} * \mathcal{I}(v(\tau,\mathbf{x}, \mathbf{z}))| \\
 &\leq \sup_{\tau, \mathbf{x}, \mathbf{z}} \E^{\mathbf{x}, \mathbf{z}} \left[\widehat{X}^u(\tau) \left | f(u) - f(v)\right | (\tau,\mathbf{x}, \mathbf{z})\right] / \mathcal{V}(0) \\
 & \quad + \widehat{M} \sup_{\tau, \mathbf{x}, \mathbf{z}} \E \left |\widehat{X}^u(\tau) - \widehat{X}^v(\tau) \right|\\
& \leq L \sup_{\tau, \mathbf{x}, \mathbf{z}} \E^{\mathbf{x}, \mathbf{z}} \left[\widehat{X}^u(\tau) |u - v|(\tau,\mathbf{x}, \mathbf{z}) \right] / \mathcal{V}(0) + \widehat{M} \|u - v\|_{\mathcal U}\\
& \leq \left( \lambda L + \widehat{M} \right) \|u - v\|_{\mathcal U},
\end{aligned}
\]
where the second inequality is derived from the local Lipschitz continuity of $f$. $L := \sup_{u \in [0,1]} |f'(u)| = (1-\delta) \overline{e^c} / (1-(1-\delta) \overline{e^{c}})^2 $
and
\begin{equation}
\label{eq:lambda}
\begin{split}
    \lambda &:= \sup_{\tau, \mathbf{x}, \mathbf{z}}  \E^{\mathbf{x}, \mathbf{z}} \left[\widehat{X}(\tau) \right]/\mathcal{V}(0) = \sup_{\tau, \mathbf{x}, \mathbf{z}}  \frac{\E^{\mathbf{x}, \mathbf{z}} \left[\widehat{X}(\tau) \right]}{X(0)} \frac{X(0)}{\mathcal{V}(0)} \\
    & < \frac{X(0)}{\mathcal{V}(0)} = \frac{1-(1-\delta)\overline{e^{c}}}{\delta }.
    \end{split}
\end{equation}
The inequality in \eqref{eq:lambda} follows from the supermartingale property $\E \left[\widehat{X}(t) \right] < x_0$ in Proposition~\ref{f1}.


Combining these quantities, we get
\begin{equation}
\label{lambdal}
\begin{aligned}
    \lambda L & \leq  \frac{1}{ \delta } \frac{(1-\delta) \overline{e^{c}}}{1-(1-\delta)\overline{e^{c}}}. \\
\end{aligned}    
\end{equation}
\item To show the contraction property of the operator $ \mathcal{F} * \mathcal{I}$ on $ \mathcal{U} := C([0,T] \times \mathbb{R}_+ \times \mathbb{R}_+))$, we need $\lambda L + \widehat{M} < 1$ to hold. This is equivalent to the following conditions 
\[
 \overline{e^{c}} (1-\delta)< 1 , 
\]
and
\begin{equation}
\label{eq: cond_M}
    \frac{1}{\delta} \frac{(1-\delta) \overline{e^{c}}}{1-(1-\delta)\overline{e^{c}}} + \frac{M}{\delta x_0}< 1.
\end{equation}


Then by the Banach fixed point theorem, under \eqref{eq: condforn}, the solution $u$ in \eqref{ufixed} is unique. Thus, \eqref{eq: condforn} is a sufficient condition that every participant achieves the unique Nash equilibrium in the sense of Definition~\ref{mfemu}. 
\end{enumerate}
\end{proof}

\subsection{General interacting particle systems}
\label{appendix: nusystem}

In the following, $|\cdot|$ denotes the Euclidean norm of vector $\mathbb{R}^d$ and the Frobenius norm of matrix $\mathbb{R}^{d\times n}$, $d =1$ or $n$ in particular. 
\begin{asmp}
\label{asmp: lipfn}
Assume the Lipschitz continuity and linear growth condition are satisfied with Borel measurable mappings $b_{i}(x, \nu)$, $s_{ik}(x, \nu)$ in \eqref{mcvsde}-\eqref{mcvsde2} from $ \mathbb{R}^n_+ \times \mathcal{P}_2(\mathbb{R}_+)$ to $\mathbb{R}^n$. That is, there exist some constants $C_L \in (0, \infty)$ and $C_G \in (0, \infty)$, such that
\begin{equation}
\label{bslip2}
    |b(x, \nu) - b(\widetilde{x}, \widetilde{\nu})| + |s(x, \nu) - s(\widetilde{x}, \widetilde{\nu})| \leq C_L[|x - \widetilde{x}| + \mathcal{W}_2(\nu, \widetilde{\nu})] , 
\end{equation}
\[
|x \beta(x, \nu)| + |x \sigma(x, \nu)| \leq C_G(1+|x|+M_2(\nu)),
\]
where 
\[
M_2(\nu) := \bigg(\int_{C([0,T], \mathbb{R}_+ \times \mathbb{A})} |x|^2 d \nu(x) \bigg)^{1/2}; \quad \nu \in \mathcal P_{2}(\mathbb R_{+} \times \mathbb A) . 
\]

Also, assume that there exists a positive constant $M_v$ such that   
for any $(x, \nu)$ in $ \mathbb{R}^n_+ \times \mathcal{P}_2(\mathbb{R}_+)$ to $\mathbb{R}^n$,
\[
|v^{} \beta( x, \nu)| + |v^{} \sigma( x, \nu)| \leq M_v,
\]
\[
|v^{} \beta( x, \nu)| + |v\sigma(x, \nu)| \leq  C_G(1+|x|+|v|+M_2(\nu)),
\]
for every $v, \widetilde{v} \in \mathbb{R}_+$, $x, \widetilde{x} \in \mathbb R_{+}^n$; $\nu, \widetilde{\nu} \in \mathcal P_{2}( \mathbb{R}_+)$. {
We consider the optimal strategy to be of the form $\phi(x, v, \nu,c)$ and assume that there exists
$L_\phi>0$ such that for all $(x, v, \nu),(x', v', \nu')$ and all $c$,
\[
  \|\phi(x, v, \nu,c)-\phi(x', v', \nu',c)\|
  \le L_\phi\big( |x-x'| + |v-v'| + \mathcal W_2(\nu,\nu') \big).
\]
}
\end{asmp}

Analogous to \cite[Theorem C.1]{ichiba2020relative}, we can derive the following result.
\begin{thm}
\label{eumckean}


Assume that initial capitalization and wealth $(\mathbf{x}_0, v_0)$ is independent of the Brownian motion $B (\cdot)$ on $(\Omega, \mathcal F, \mathbb F, \mathbb P)$, and follows $\E[\sup_{0 \leq t \leq T} \|(\mathbf{x}_0, v_0)\|^2] \leq \infty$.
Under Assumptions~\ref{asmp: lipfn}, the McKean-Vlasov system \eqref{mcvsde}-\eqref{mcvsde2} admits a unique strong solution.
\end{thm}

\subsection{Approximation and convergence results}
\label{appendix: proof5}

\begin{proof}[Proof of Proposition~\ref{xlimit}]
First, we show that 
$\mathbb{P} \circ (\mathcal{X}^N)^{-1}$ is tight. 
A sequence of measures $\mathbb{P} \circ (\mathcal{X}^N)^{-1}$ on $\mathcal{P}_2(C([0,T]; \mathbb{R}_+))$ is tight if and only if both conditions are satisfied.
\begin{itemize}
    \item $\lim_{N \rightarrow \infty} \mathbb{P}^{\mathcal{X}^N}(|\mathcal{X}^N(0)| \geq a) = 0;$
    \item $\lim_{\delta \rightarrow 0} \limsup_{N \rightarrow \infty} \mathbb{P}^{\mathcal{X}^N}(\{\sup_{|s-t|\leq \delta}|\mathcal{X}^N(s) - \mathcal{X}^N(t)| > \epsilon\}) = 0.$
\end{itemize}
By \cite[Theorem 4.3]{ikeda2014stochastic}, it is sufficient to show that
\begin{enumerate}[label=(\roman*)]
    \item \label{tightcond1} there exist positive constants $M_x$ and $\gamma$ such that $\E\{|\mathbf{x}^N|^{\gamma}\} \leq M_x$ for every $N=1,2, \ldots$,
    \item \label{tightcond2} there exist positive constants $M_k$ and $\delta_1$, $\delta_2$ such that $\E\{|\mathcal{X}^N(t) - \mathcal{X}^N(s)|^{\delta_1}\} \leq M_k |t-s|^{1+\delta_2}$ for every $N$, $t,s \in [0,k]$, $k = 1, 2, \ldots$.
\end{enumerate}

With $\mathbf{x} \in L^2(\Omega, \mathcal{F}_0, \mathbb{P}; \mathbb{R}^n_+)$ in Assumption~\ref{xv}, condition \ref{tightcond1} in the above statement holds.

To check condition \ref{tightcond2}, by Cauchy–Schwarz inequality,
\[
\begin{aligned}
|\mathcal{X}^N(t) - \mathcal{X}^N(s)|^{2} 
 = & |X^N_1(t) - X^N_1(s)|^{2} + \ldots + |X^N_n(t) - X^N_n(s)|^{2} \\
 = & \sum_{i=1}^n \Big| \int_s^t X^N_i(r) \beta_i(r)dr + \sum_{k=1}^n \int_s^t X^N_i(r) \sigma_{ik}(r)dW_k(r) \Big|^{2} \\
\end{aligned}
\]
Then, with the Lipschitz and linear growth conditions, we first derive that the drift term satisfies
\begin{align*}
    \E \left[ \left| \int_s^t b_i( \mathcal X^N(u), \nu^N_u) du \right|^{\delta_1} \right] &\leq (t-s)^{\delta_1} C_G^{\delta_1} \sup_{t \in [0,T]} (\E[|\mathcal X^N(t)| + M_2(\nu^N_t)] + 1)^{\delta_1}.
\end{align*}
For the diffusion term, by It\^o's isometry, it holds
\begin{align*}
    & \E \left[ \left| \int_s^t s_i( \mathcal X^N(u), \nu^N_u) dW^N(u) \right|^{\delta_1} \right] \\
    &\leq (t-s)^{\frac{\delta_1}{2}} C_G^{\delta_1} \sup_{t \in [0,T]} (\E[|\mathcal X^N(t)| + M_2(\nu^N_t)] + 1)^{\delta_1}.
\end{align*}
Combining the bounds for the drift and diffusion terms, we have:
\begin{align*}
\E\{|\mathcal{X}^N(t) - \mathcal{X}^N(s)|^{\delta_1}\} &\leq \sum_{i=1}^n  \E \left[ \left| \int_s^t b_i(\mathcal X^N(u), \nu^N_u) du \right|^{\delta_1} \right]  \\
& \qquad {} + \sum_{i=1}^n  \E \left[ \left| \int_s^t s_i(\mathcal X^N(u), \nu^N_u) dW^N(u) \right|^{\delta_1} \right] \\
&\leq C_N |t-s|^{1 + \delta_2}
\end{align*}
where $\delta_2 = \frac{\delta_1}{2} - 1$, $\delta_1 > 2$, and
\[ C_N = n C_G^{\delta_1} \E[|\mathcal X^N(t)| + M_2(\nu^N_t) +1]^{\delta_1}. \]
Thus, condition \ref{tightcond2} follows.

By Prokhorov theorem \cite{billingsley2013convergence} {and the uniformly bounded second moments of $ \{\mathcal X^{N,\ell}\}_{N \in \mathbb{N}}$, }tightness implies relative compactness in $\mathcal P_2(C([0,T];\mathbb R_+^n))$, which means here that each subsequence of $\{ \mathcal{X}^N \}$ contains a further subsequence converging weakly on the space $C([0,T]; \mathbb{R}^n_+)$. As a result, a subsequence exists such that $\mathcal{X}(t) = \lim_{N \rightarrow \infty} \mathcal{X}^N(t)$ a.s. 
The limit of $\{\mathbb{P} \circ (\mathcal{X}^N)^{-1}\}$ is unique and hence $\{\mathbb{P} \circ (\mathcal{X}^N)^{-1}\}$ converges weakly to $\{\mathbb{P} \circ (\mathcal{X})^{-1}\}$ as $N \rightarrow \infty$. 
\end{proof}

\begin{proof}[Proof of Proposition~\ref{prop: tightv}]
By \eqref{eq: nv} with $V^{N, \ell}(0) = v^{\ell}$, it follows that
\begin{equation}
\begin{split}
 |V^{N, \ell}(t)|^2 & \leq 3 \left(v^{\ell}\right)^2 + \underbrace{\, 3 \left( \int_0^t V^{N, \ell}(s) \sum_{i=1}^{n}\pi^{\ell}_i(s) \beta_{i}(\mathcal X(s), \nu_s)ds \right)^2}_{(i)} \\
 & \quad {} + \underbrace{\, 3 \left( \int_0^t V^{N, \ell}(s) \sum_{i=1}^{n} \sum_{k=1}^{n} \pi^{\ell}_i(s) \sigma_{ik}(\mathcal X(s), \nu_s)dW^N_{k}(s) \right)^2}_{(ii)}.
\end{split}
\end{equation}
By the linear growth condition in Assumption~\ref{xv}, 
\begin{align*}
\E[(i)] & \leq 3 n t \sum_{i=1}^{n} \E \int_0^t \left( V^{N, \ell}(s)  \pi_i^{\ell}(s) \beta_{i}(\mathcal X(s), \nu_s)\right)^2 ds \\
& \leq 3 C_G n t \E \int_0^t  \left( V^{N,\ell}(s) \right)^2 + |\mathcal X^N(s)|^2 + M_2^2(\nu^N_s) +1 ds \\
& \leq 3 C_G n t \left( \E \left( \sup_{0 \leq s \leq t}|\mathcal X^N(s)|^2 + \sup_{0 \leq s \leq t} M_2^2(\nu^N_s) +1\right) + \E \int_0^t  \left( V^{N,\ell}(s) \right)^2 ds \right).
\end{align*}
Next, by It\^o's isometry,
\begin{align*}
\E[(ii)] & \leq 3 n \sum_{i=1}^{n} \E \int_0^t  \left( V^{N, \ell}(s) \sum_{k=1}^{n} \pi^{\ell}_i(s) \sigma_{ik}(s, \mathcal X(s), \mu_s) \right)^2 ds\\
& \leq 3 C_G n \left( \E \left( \sup_{0 \leq s \leq t}|\mathcal X^N(s)|^2 + \sup_{0 \leq s \leq t} M_2^2(\nu^N_s) +1\right) + \E \int_0^t  \left( V^{N,\ell}(s) \right)^2 ds \right).
\end{align*}

By Gr\"onwall's inequality, we can derive the boundedness of the moment of $V^{N,\ell}(t)$ that for any fixed $ t \in [0,T]$,
\[ \E |V^{N,\ell}(t)|^2 \leq 3 \left( (v^{\ell})^2 + C_G n t \sup_{0 \leq s \leq t} \E \left(|\mathcal X^N(s)|^2 + M_2^2(\nu^N_s) +1\right) \right) \exp C_t, \]
where $C_t = 3 C_G n t (t+1).$

Next, we derive the equicontinuity. For every $\epsilon$, we consider $|t-s|<\delta$, given $\delta >0$, $t, s \in [0,T]$. By using Ito's formula and Lipschitz conditions, it follows that
\begin{align*}
\E \left[ |V^{N,\ell}(t) - V^{N,\ell}(s)|^2 \right] &\leq 2\delta C_G^2 \sup_{u \in [s,t]}  \E\left[ \left(|\mathcal X^N(u)| + M_2(\nu^N_u) + 1\right)^2 \right]
\end{align*}
By selecting $ \delta $ sufficiently small, we can ensure the equicontinuity from Chebyshev's inequality.
Now we can use Arzel\`a-Ascoli theorem to conclude that the family of functions $\{ V^{N,\ell}\}_{N \in \mathbb{N}} $ for each $\ell$ is equicontinuous and uniformly bounded, then it is pre-compact in the space of continuous functions. This means every sequence has a uniformly convergent subsequence. Similar to \cite[Vol II, Lemma 3.16]{carmona2015probabilistic}, we can show the tightness of $\mu^N$. Given $ \epsilon > 0 $, by the tightness of the law of $ \{V^{N,\ell}\}_{N \in \mathbb{N}}$, there exists a sequence of compact sets $\{K_p\}_{p \in \mathbb{N}}$, where $ K_p \subset \mathbb R_+ $, such that
\[ \sup_{N \in \mathbb{N}} P(V^{N,\ell} \notin K_p) \leq \frac{\epsilon}{2^p}, \quad \text{for each } \ell = 1, \ldots, N, \quad p \in \mathbb{N}. \]
Consider the set $ B = \mathbb R_+ \setminus \cup_{p=1}^{\infty} K_p$. Thus
\begin{align*}
\E[\mu^N(B)] = \E\left[ \frac{1}{N} \sum_{k=1}^{N} 1_{\{V^{N,\ell} \in B\}} \right] = \frac{1}{N} \sum_{k=1}^{N} P(V^{N,\ell} \notin \cup_{p=1}^{\infty} K_p) \leq \frac{1}{N} \sum_{\ell=1}^N \sum_{p=1}^{\infty} \frac{\epsilon}{2^p} = \epsilon.
\end{align*}

{Thus, by the uniformly bounded second moments of $ \{V^{N,\ell}\}_{N \in \mathbb{N}}$}, $\mu_N$ is tight on $\mathcal{P}_2(C([0,T]; \mathbb{R}_+))$ equipped with the weak topology. Since $\pi \in \mathbb A$ is a compact space, we can show similarly the tightness of $\nu_N$ on $\mathcal{P}_2(C([0,T]; \mathbb{R}_+))$ equipped with the Wasserstein distance.

\end{proof}

\begin{proof}[Proof of Theorem~\ref{tight}]
Let the wealth process $(\Bar{V}^{\ell})$ be 
the solution of \eqref{mcvsde2} with closed loop Markovian strategy $\phi : [0,T] \times \mathcal{P}_2(\mathbb{R}_+) \times \mathbb{R} \rightarrow \mathbb{A}$ is a bounded function such that at Nash equilibrium, the strategy $\pi^{\star}(t) = \phi(\mathcal{X}(t), \nu_t, c)$. We can define the finite-player counterparts $(V^{\ell,N}, \pi^{\ell,N})$, where $\pi^{\ell,N} = \phi(\mathcal{X}^N(t), \nu^N_t, c^{\ell})$. 
For $\ell = 1, \ldots, N$, due to the absence of independent noise, the randomness of $(\Bar{V}^{\ell}, \pi^{\ell \star})$, comes all from the initial condition if conditioned on the common noise $B$, where $\pi^{\ell \star} := \phi(\mathcal{X}(t), \nu_t, c^{\ell})$.
Each investor takes an i.i.d sample of $c_{\ell}$ in \eqref{2}, so $(\Bar{V}^{\ell})$ is conditionally i.i.d. given common noises $B$. 
In the same vein, \eqref{anotherpi2} indicates that optimal strategies at equilibrium are conditionally i.i.d. given common noises $B$. Thus, $(U^1, \pi^1)$ from a representative player is exchangeable from other players, given $\mathbb{F}^{B}$. Without loss of generality, we can discuss the problem with a representative player. 

By Theorem~\ref{thm:unique}, at Nash equilibrium, the optimal arbitrage quantity $u^{\star}$ is uniquely determined. Thus, $\Bar{V}^{\ell}(0) = v^{\ell} = u^{\star}(T) \mathcal{V}(0) \sim \mu^\star_0,$ where $\mathcal{V}(0)$ follows \eqref{v0v0}. $\hat V^{N, \ell}(0) = v^{N, \ell}$ is a random sample from $\mu^\star_0$ for each $\ell$.  For $t > 0$,
\begin{equation}
\label{ul1}
d \Bar{V}^{\ell}(t) = \Bar{V}^{\ell}(t) \phi(\mathcal{X}(t), \nu_t, c^{\ell}) \beta(X(t), \nu_t)dt + \Bar{V}^{\ell}(t) \phi(\mathcal{X}(t), \nu_t, c^{\ell}) \sigma(\mathcal{X}(t), \nu_t)dB_t, \\
\end{equation}
\begin{equation}
\label{vl1}
\begin{aligned}
    d \hat V^{N, \ell}(t) = \hat V^{N, \ell}(t) \phi(\mathcal{X}^N(t), \nu^N_t, & c^{\ell}) \beta(\mathcal{X}^N(t), \nu^N_t)dt \\
    & + \hat V^{N, \ell}(t) \phi(\mathcal{X}^N(t), \nu^N_t, c^{\ell}) \sigma(\mathcal{X}^N(t), \nu^N_t)dW^N_t.
\end{aligned}
\end{equation}
Let us denote the supremum norm $\|x\|_t = \sup_{s \in [0,t]} |x_s|$. We use a similar approach as \cite[Theorem 3.3]{lacker2018mean} to define a truncated Wasserstein distance on $\mathcal{P}_2(\mathcal{C}^A)$ by
\[
d_t^2(\nu, \nu') := \inf_{\pi \in \Pi(\nu, \nu')} \int_{\mathcal{C}^A \times \mathcal{C}^A} \|x-x'\|_t^2 \pi(dx, dx').
\]
Consider $\mathcal G_t := \mathcal F_t^B \vee \mathcal \sigma({W^N_s: s \leq t, N \in \mathbb{N}})$. We have
\begin{equation}
\label{1}
\begin{aligned}
        & \E[\|\hat V^{N, \ell}(\mathcal{X}^N, \nu^N) - \Bar{V}^{\ell}(\mathcal{X}, \nu)\|_t^2 | \mathcal G_t] \\
= &\, \E \left[\int_0^t  \|\hat{V}^{N, \ell}(\mathcal{X}^N_r, \nu^N_r)\phi(\mathcal{X}^N(t), \nu^N_t, c^{\ell}) \beta(\mathcal{X}^N(t), \nu^N_t)- \Bar{V}^{\ell}(\mathcal{X}(r), \nu_r) \phi(\mathcal{X}(t), \nu_t, c^{\ell}) \beta(X(t), \nu_t\|^2 dr \Big| \mathcal G_t\right]\\
\leq &\, \E \left[C_2 \int_0^t  \|\mathcal{X}^N_r - \mathcal{X}_r\|^2 dr + C_2 \int_0^t d_r^2(\nu^N, \nu) dr 
        + n C_2 \int_0^t \|\hat V^{N, \ell}_r - \Bar{V}^{\ell}_r\|^2 dr\Big| \mathcal G_t\right]
        \\
        \leq & \, 
        C_F  \E \left[\int_0^t d_r^2(\nu^N, \nu) dr\Big| \mathcal G_t\right],
\end{aligned}
\end{equation}
where the first equation is a consequence of It\^o's isometry and the martingale property, the first inequality is a consequence of It\^o's isometry, the Lipschitz continuity of the coefficients, and the boundedness of strategies; the second inequality is derived by applications of Gr\"onwall's Lemma~\ref{gronwall} on the first difference term about $\mathcal{X}^N, \mathcal{X}$ and third difference term about $\hat V^{N, \ell}, \Bar{V}^{\ell}$ in the first inequality. $C_G$ is the constant determined by the Gr\"onwall's Lemma~\ref{gronwall}.
We follow the coupling arguments in \cite{carmona2018probabilistic}, the empirical measure of $(\hat V^{N, \ell}, \Bar{V}^{\ell})$ is a coupling of the empirical measure of $\mu^N$ and $\Tilde{\mu}^N$, where $\Tilde{\mu}^N$ is the empirical measure of $N$ i.i.d. samples $\Bar{V}^{\ell}$ in \eqref{ul1}. Thus
\begin{equation}
\label{2}
    d_t^2(\mu^N, \Tilde{\mu}^N) \leq \frac{1}{N} \sum_{\ell=1}^N \|\hat V^{N, \ell} - \Bar{V}^{\ell}\|^2_t, \quad \text{a.s.}.
\end{equation}

By the triangle inequality and \eqref{1}-\eqref{2},
\[
\begin{aligned}
    \E[d_t^2(\mu^N, \mu)] &\leq 2 \E[d_t^2(\Tilde{\mu}^N, \mu)] + 2C_F\E[\int_0^t d_r^2(\nu^N, \nu) dr]\\
    &\leq 2 \E[d_t^2(\Tilde{\mu}^N, \mu)] + 2\widetilde{C}_F\E[\int_0^t d_r^2(\mu^N, \mu) dr].
\end{aligned}
\]
By Gr\"onwall's inequality and set $t=T$, under the initial conditions of capitalization and preferences, it follows
\begin{equation}
\label{ineqc}
    \E[\mathcal{W}_2^2(\mu^N, \mu)] \leq 2e^{2\widetilde{C}_F T} \E[\mathcal{W}_2^2(\Tilde{\mu}^N, \mu)].
\end{equation}


Since in the empirical measure $\Tilde{\mu}^N$, $\Bar{V}^{\ell}$ is i.i.d given the common filtration $\mathcal{F}^{B}$, thus under the exchangeability of $\Bar{V}^{\ell}$, {we use conditional law of large numbers (See \cite[Theorem 3.5]{majerek2005conditional}) to get that there exists $\Bar{V}$ such that
\[
\lim_{N \rightarrow \infty} \frac{1}{N} \sum_{\ell=1}^N f(\Bar{V}^{\ell}) = \E[f(\Bar{V})|\mathcal F^B], \quad \text{a.s., \ for every} f \in C_b \left(\mathbb{R}^n \right).
\]
}
We then apply \cite[Theorem 6.6]{parthasarathy2005probability}, which gives that on a separable metric space, $\Tilde{\mu}^N \rightarrow \mu$ weakly,
\[
\lim_{N \rightarrow \infty} \int_{\mathbb{R}^N} d(x,x_0)^2 \Tilde{\mu}^N(dx) = \int_{\mathbb{R}^N} d(x,x_0)^2 \mu(dx) \quad \text{a.s.}
\]

Therefore, $\E[\mathcal{W}_2^2(\Tilde{\mu}^N, \mu)] \rightarrow 0.$ By \eqref{ineqc}, it leads to $\E[\mathcal{W}_2^2(\mu^N, \mu)] \rightarrow 0$. We can use similar methods to derive $\E[\mathcal{W}_2^2(\nu^N, \nu)] \rightarrow 0$ under Assumption~\ref{xv}, by using the boundedness of $\pi(\cdot)$, and Gr\"onwall's inequality (Lemma~\ref{gronwall}).
\end{proof}

\begin{proof}[Proof of Theorem~\ref{netomfe}]
The first step is to construct the fixed point operator $\mathfrak G^N$ of the $N$-player game analogously to the mean field game setup. We can then show that the fixed point operator $\mathfrak G^N : \mathcal U \rightarrow \mathcal U$ in \eqref{ufixedN} converge uniformly to $\mathfrak G: \mathcal U \rightarrow \mathcal U$ in \eqref{ufixed1}.

\begin{enumerate}

\item \textbf{Weak convergence of the optimal SDE systems.}

{By Proposition~\ref{xlimit}-\ref{prop: tightv}, $\mathbb{P} \circ (\mathcal{X}^{N}, V^{N}, W^N)$ is tight on the space $C([0,T]; \mathbb{R}_+^n) \times C([0,T]; \mathbb{R}_+^N) \times \mathcal{P}_2(C([0,T], \mathbb{R}_+ \times \mathbb{A})) \times C([0,T]; \mathbb{R}^n)$. Consider the optimal strategy at equilibrium, by Prokhorov's theorem, every sequence $\mathbb{P} \circ(\mathcal X^{N \star}, V^{N \star}, \mu^{N \star}, \nu^{N \star})^{-1}$ admits a weakly convergent subsequence $\mathbb{P} \circ(\mathcal X^{N_k}, V^{N_k}, \mu^{N_k}, \nu^{N_k})^{-1}$. By the propagation of chaos result Theorem~\ref{tight}, there exists the limit $\mathbb{P} \circ(\hat{\mathcal X}, \hat{V}, \mu^{\star}, \nu^{\star})^{-1}$ of the convergent subsequence, and it solves a martingale problem, where the limiting martingale problem coincides with that of \eqref{mcvsde}-\eqref{mcvsde2}, evaluated at $(\hat{\mathcal X}, \hat{V}, \mu^{\star}, \nu^{\star})$. 
Since the solution of \eqref{mcvsde}-\eqref{mcvsde2} is unique in law, we conclude that $\mathbb{P} \circ (\mathcal{X}^{N \star}, V^{N \star}, W^N)$ converges weakly to $\mathbb{P} \circ (\mathcal{X}^{\star}, V^{\star}, B)$. }

    \item \textbf{Construction of $\mathfrak G^N$.} Next, we construct the fixed point problem that solves the aforementioned $(\mathcal{X}^{N \star}, V^{N \star})$. We set $\mathfrak G^N : \mathcal U \rightarrow \mathcal U$, in analogy to $\mathfrak G$ in Theorem~\ref{thm:unique}. Since in $N$-player game, $u^{\ell}(T-t,\mathbf{x}, \mathbf{z})  = e^{c_{\ell}} U^N(T-t,\mathbf{x}, \mathbf{z})$, and
\begin{equation}
\label{ute2N}
    \begin{split}
    U^N(T-t,\mathbf{x}, \mathbf{z}) 
    & = \frac{\delta}{\mathcal{V}(0)} \E \left[ \frac{\widehat{X}^{N \star}(T-t)}{1- \frac{1-\delta}{N}  U^N(t) \sum_{k=1}^N e^{c_{k}}} \Big | \mathcal{F}_t^{MF}\right].
    \end{split}
\end{equation}
Investor $\ell$ looks for the best response by fixing the states of the other $N-1$ players when looking for Nash equilibrium. Thus, for simplicity, 
we denote $U^N(T-t,\mathbf{x}, \mathbf{z})$ as $U^{N \star}_{T-t}$, and the fixed relative arbitrage quantity for the other players at time $t \in [0,T]$ as $\widetilde{U}^{\star}(t) := U^{\star}(t) \sum_{k \neq \ell} e^{c_{k}}$. Define $f(u) = \frac{1}{1- \frac{1-\delta}{N}  u \sum_{k=1}^N e^{c_{k}}} $. Define an operator $\mathfrak G^N: \mathcal{U} \rightarrow \mathcal{U}$ by
\begin{equation}
 \label{ufixed0N}
 \begin{aligned}
\left[ \mathfrak G^N * U \right](T-t, \mathbf{x}, \mathbf{z})&= \frac{1}{\mathcal{V}^N(0)} \E \left[\widehat{X}^{N \star}(T-t) \frac{1}{1- \frac{1-\delta}{N}  \left [e^{c_{\ell}} U(t) + \widetilde{U}^{\star}_t \right] } \Big | \mathcal{F}_t^{MF} \right] \\
&= \frac{1}{\mathcal{V}^N(0)} \E \left[ \widehat{X}^{N \star}(T-t) f (U(t)) | \mathcal{F}_t^{MF}\right],
 \end{aligned}
\end{equation}
where $U := \left(U(t)\right)_{t \in [0,T]} \in \mathcal{U}$. Hence, $\mathfrak G^N *\mathcal{I}(U)(t) = \mathfrak G^N(U)(T-t)$. At the Nash equilibrium, \eqref{ute2N} leads to the fixed-point condition
\begin{equation}
\label{ufixedN}
\left[ \mathfrak G^N * \mathcal{I}( U^N) \right](T-t,\mathbf{x}, \mathbf{z}) = U^N(T-t,\mathbf{x}, \mathbf{z}).
\end{equation}

    \item \textbf{Convergence of $\mathfrak G^N$.}
{For given $(\mathbf{x}, \mathbf{z})$, $\mathcal{V}(0)$ and $\mathcal{V}^N(0)$ are the same by definition. For any $t \in [0,T]$, $\widehat{X}^{N \star}(T-t) f (U^N(t))$ is bounded in  $L^2(\Omega, \mathcal F, \mathbb P)$ as $\widehat{X}^{N \star}(T-t)$ is $L^2$ bounded by Assumption~\ref{xv} and Proposition~\ref{eumckean-1}, and $f (U^N(t))$ is bounded from \eqref{ufixed0N},
for all $ N \in \mathbb{N}$. Then, for any $t \in [0,T]$, the family $ \left\{ \widehat{X}^{N \star}(T-t) f (U^N(t)) \right\}_{N \in \mathbb N} $ is uniformly integrable by the criterion of de la Vall\'ee Poussin (choose a quadratic function as the convex function in the criterion).}{Hence, by rewriting the conditional expectation using continuous test functions, and by the aforementioned joint law convergence and the continuous mapping theorem, the operator \eqref{ufixed0N} follows
\[ 
\begin{aligned}
\lim_{n\rightarrow\infty} \E^{\mathbf{x}, \mathbf{z}} \left[\widehat{X}^{N \star}(T-t) f (U^N(t)) | \mathcal F^{MF}_t \right] &= \E^{\mathbf{x}, \mathbf{z}} \left[\lim_{n\rightarrow\infty} \widehat{X}^{N \star}(T-t) f (U^N(t)) |\mathcal F^{MF}_t \right]\\
&= \E^{\mathbf{x}, \mathbf{z}} \left[\widehat{X}^\star(T-t) f (U(t)) |\mathcal F^{MF}_t \right],     
\end{aligned}
\]
for each $ t \in [0,T]$, $U$ is the mean-field counterpart defined in \eqref{eq: ufixedt}. Note that the strategy $\pi(\cdot)$ defined implicitly in $\widehat{X}(T-t) f (U(t)) $ and $\widehat{X}^N(T-t) f (U^N_t)$ are the same.}

\item \textbf{Convergence of the fixed point solution.} By Theorem~\ref{thm:unique}, $\mathfrak G $ has a unique fixed point $ U \in \mathcal U$. That is, for any $ \widetilde U \in \mathcal{U} $ that is not $ U $, $ \|\mathfrak G(\widetilde U) - \widetilde U\|_{\mathcal U} > \delta $ for some positive $ \delta $. Similarly, $U^N $ is the unique fixed point of $ \mathfrak G^N $ for any $ N $. Hence
\begin{equation}
\label{fixedpt_fg}
    \mathfrak G^N(U^N) = U^N; \quad \mathfrak G(U) = U
\end{equation}
Essentially, we aim to show that the fixed point $U^N $ converges to $ U $ as $ N \to \infty $.

As shown in the previous step, $ \mathfrak G^N $ converges uniformly to $\mathfrak G $ on $\mathcal U$. We have, for any $ \epsilon > 0 $, there exists $ N_0 $ such that for all $ N > N_0 $ and for all $ x  \in [0,1]$, $ \|\mathfrak G^N(x) - \mathfrak G(x)\|_{\mathcal U} < \epsilon $. We next show the fixed point property of $U^N$ through a contradiction argument. Suppose a fixed point solution $\|U^N - U\|_{\mathcal U} > 0$. Then, for $ N > N_0$, $\|\mathfrak G^N(U^N) -U^N\|_{\mathcal U} = \|\mathfrak G^N(U^N) - \mathfrak G(U^N) + \mathfrak G(U^N) -U^N\|_{\mathcal U} \geq \left| \|\mathfrak G(U^N) -U^N\|_{\mathcal U} - \|\mathfrak G^N(U^N) - \mathfrak G(U^N) \|_{\mathcal U} \right| >0$, by the triangle inequality.
This is contradicted with the definition of $U^N$ as the fixed point solution of \eqref{fixedpt_fg}.
\end{enumerate}

\end{proof}

\end{appendices}


\newpage

\begin{thebibliography}{999}

\bibitem{billingsley2013convergence}
Patrick Billingsley,
\emph{Convergence of probability measures}.
John Wiley \& Sons, New York, 2013.

\bibitem{campbell2022functional}
Steven Campbell and Ting-Kam Leonard Wong,
\emph{Functional portfolio optimization in stochastic portfolio theory}.
SIAM Journal on Financial Mathematics, 13(2):576--618, 2022.

\bibitem{cardaliaguet2022first}
Pierre Cardaliaguet and Panagiotis E. Souganidis,
\emph{On first order mean field game systems with a common noise}.
The Annals of Applied Probability, 32(3):2289--2326, 2022.

\bibitem{carmona2016lectures}
Ren{\'e} Carmona,
\emph{Lectures on BSDEs, stochastic control, and stochastic differential games with financial applications}.
SIAM, 2016.

\bibitem{carmona2015probabilistic}
Ren{\'e} Carmona and Daniel Lacker,
\emph{A probabilistic weak formulation of mean field games and applications}.
Annals of Applied Probability, 25(3):1189--1231, 2015.

\bibitem{carmona2016mean}
Ren{\'e} Carmona, Fran{\c{c}}ois Delarue, and Daniel Lacker,
\emph{Mean field games with common noise}.
Annals of Probability, 44(6):3740--3803, 2016.

\bibitem{carmona2018probabilistic}
Ren{\'e} Carmona, Fran{\c{c}}ois Delarue, et al.,
\emph{Probabilistic theory of mean field games with applications I--II}.
Springer, New York, 2018.

\bibitem{cover1991universal}
Thomas M. Cover,
\emph{Universal portfolios}.
Mathematical Finance, 1(1):1--29, 1991.

\bibitem{cuchiero2019cover}
Christa Cuchiero, Walter Schachermayer, and Ting-Kam Leonard Wong,
\emph{Cover's universal portfolio, stochastic portfolio theory, and the num{\'e}raire portfolio}.
Mathematical Finance, 29(3):773--803, 2019.

\bibitem{djete2022extended}
Mao Fabrice Djete,
\emph{Extended mean field control problem: a propagation of chaos result}.
Electronic Journal of Probability, 27:1--53, 2022.

\bibitem{djete2023mean}
Mao Fabrice Djete,
\emph{Mean field games of controls: on the convergence of Nash equilibria}.
The Annals of Applied Probability, 33(4):2824--2862, 2023.

\bibitem{fernholz2010optimal}
Daniel Fernholz and Ioannis Karatzas,
\emph{On optimal arbitrage}.
Annals of Applied Probability, 20(4):1179--1204, 2010.

\bibitem{fernholz2011optimal}
Daniel Fernholz and Ioannis Karatzas,
\emph{Optimal arbitrage under model uncertainty}.
Annals of Applied Probability, 21(6):2191--2225, 2011.

\bibitem{fernholz2002stochastic}
E. Robert Fernholz,
\emph{Stochastic portfolio theory}.
Springer, New York, 2002.

\bibitem{fernholz1999portfolio}
Robert Fernholz,
\emph{Portfolio generating functions}.
In \emph{Quantitative Analysis in Financial Markets: Collected Papers of the New York University Mathematical Finance Seminar}, pages 344--367. World Scientific, 1999.

\bibitem{fernholz2005relative}
Robert Fernholz and Ioannis Karatzas,
\emph{Relative arbitrage in volatility-stabilized markets}.
Annals of Finance, 1:149--177, 2005.

\bibitem{follmer1999quantile}
H. F\"ollmer and P. Leukert, 
\emph{Quantile hedging}. Finance and Stochastics, 3(3):251–273, 1999.

\bibitem{friedman1975stochastic}
Avner Friedman,
\emph{Stochastic differential equations and applications}.
In \emph{Stochastic differential equations}, pages 75--148. Springer, Berlin, 2010.

\bibitem{hammersley2021weak}
William R. P. Hammersley, David {\v{S}}i{\v{s}}ka, and {\L}ukasz Szpruch,
\emph{Weak existence and uniqueness for McKean--Vlasov SDEs with common noise}.
Annals of Applied Probability, 49:527--555, 2021.

\bibitem{hu2022n}
Ruimeng Hu and Thaleia Zariphopoulou,
\emph{{$N$}-player and mean-field games in It{\^o}-diffusion markets with competitive or homophilous interaction}.
In \emph{Stochastic Analysis, Filtering, and Stochastic Optimization: A Commemorative Volume to Honor Mark H.\ A.\ Davis's Contributions}, pages 209--237. Springer, New York, 2022.

\bibitem{huang2006large}
Minyi Huang, Roland P. Malham{\'e}, and Peter E. Caines,
\emph{Large population stochastic dynamic games: closed-loop McKean--Vlasov systems and the Nash certainty equivalence principle}.
Commun.\ Inf.\ Syst., 6(3):221--252, 2006.

\bibitem{ichiba2020relative}
Tomoyuki Ichiba and Nicole Tianjiao Yang,
\emph{Relative arbitrage opportunities with interactions among $N$ investors}.
arXiv preprint arXiv:2006.15158, 2020.

\bibitem{ikeda2014stochastic}
Nobuyuki Ikeda and Shinzo Watanabe,
\emph{Stochastic differential equations and diffusion processes}.
Elsevier, North-Holland Mathematical Library, 1981.

\bibitem{karatzas2009stochastic}
Ioannis Karatzas and Robert Fernholz,
\emph{Stochastic portfolio theory: an overview}.
Handbook of Numerical Analysis, 15:89--167, 2009. Elsevier.

\bibitem{kurtz2014weak}
Thomas Kurtz,
\emph{Weak and strong solutions of general stochastic models}.
Electron.\ Commun.\ Probab., 19:1--16, 2014.

\bibitem{lacker2018mean}
Daniel Lacker,
\emph{Mean field games and interacting particle systems}.
preprint, 2018.

\bibitem{10.1214/22-AAP1876}
Daniel Lacker and Luc Le Flem,
\emph{Closed-loop convergence for mean field games with common noise}.
The Annals of Applied Probability, 33(4):2681--2733, 2023.

\bibitem{lacker2019mean}
Daniel Lacker and Thaleia Zariphopoulou,
\emph{Mean field and n-agent games for optimal investment under relative performance criteria}.
Mathematical Finance, 29(4):1003--1038, 2019.

\bibitem{lasry2007mean}
Jean-Michel Lasry and Pierre-Louis Lions,
\emph{Mean field games}.
Japanese Journal of Mathematics, 2(1):229--260, 2007.

\bibitem{majerek2005conditional}
Dariusz Majerek, Wioletta Nowak, and Wieslaw Zieba,
\emph{Conditional strong law of large number}.
Int.\ J.\ Pure Appl.\ Math., 20(2):143--156, 2005.

\bibitem{mishura2020existence}
Yuliya Mishura and Alexander Veretennikov,
\emph{Existence and uniqueness theorems for solutions of McKean--Vlasov stochastic equations}.
Theory of Probability and Mathematical Statistics, 103:59--101, 2020.

\bibitem{mou2023minimal}
Chenchen Mou and Jianfeng Zhang,
\emph{Minimal solutions of master equations for extended mean field games}.
arXiv preprint arXiv:2303.00230, 2023.

\bibitem{parthasarathy2005probability}
Kalyanapuram Rangachari Parthasarathy,
\emph{Probability measures on metric spaces}.
Volume 352. American Mathematical Soc., AMS Chelsea Publishing, 2005.

\bibitem{ruf2011optimal}
Johannes Karl Dominik Ruf,
\emph{Optimal trading strategies under arbitrage}.
Columbia University, New York, 2011.

\bibitem{sznitman1991topics}
Alain-Sol Sznitman,
\emph{Topics in propagation of chaos}.
Lecture Notes in Mathematics, pages 165--251, 1991. Springer Berlin Heidelberg.

\bibitem{villani2009optimal}
C{\'e}dric Villani,
\emph{Optimal transport: old and new}.
Volume 338. Springer, Berlin, 2009.

\bibitem{wong2015optimization}
Ting-Kam Leonard Wong,
\emph{Optimization of relative arbitrage}.
Annals of Finance, 11:345--382, 2015.

\bibitem{wong}
T.K.L. Wong,
\emph{Information geometry in portfolio theory}.
In Frank Nielsen (Ed.), Geometric Structures of Information, Springer (2019).


\bibitem{wong2015universal}
Ting-Kam Leonard Wong,
\emph{Universal portfolios in stochastic portfolio theory}.
arXiv preprint arXiv:1510.02808, 2015.

\bibitem{yang2024finding}
Nicole Tianjiao Yang and Tomoyuki Ichiba,
\emph{Finding the nonnegative minimal solutions of Cauchy PDEs in a volatility-stabilized market}.
SIAM Journal on Financial Mathematics, 16(4): SC76-SC87, 2025.

\bibitem{yang2021topics}
Tianjiao Yang,
\emph{Topics in relative arbitrage, stochastic games and high-dimensional PDEs}.
University of California, Santa Barbara, 2021.

\end{thebibliography}
\end{document}